\documentclass[journal,draftcls,onecolumn,12pt,twoside]{IEEEtran}

\usepackage[utf8]{inputenc} 
\usepackage[T1]{fontenc}
\usepackage{url}
\usepackage{ifthen}
\usepackage{cite}
\usepackage[cmex10]{amsmath}
\usepackage{amsthm}
\usepackage{cite}
\usepackage{amssymb}
\usepackage{epsfig}
\usepackage{epstopdf}
\usepackage{graphicx}
\usepackage{float}
\usepackage{caption}
\usepackage{subcaption}
\usepackage{dsfont}
\usepackage[linesnumbered,ruled,vlined]{algorithm2e}
\usepackage[svgnames]{xcolor}
\usepackage{tikz}
\usepackage{pgfplots}

\IEEEoverridecommandlockouts

\newcommand{\xb}{\boldsymbol{x}}
\newcommand{\Xb}{\boldsymbol{X}}
\newcommand{\barxb}{\bar{\boldsymbol{x}}}

\newcommand{\hatxb}{\hat{\boldsymbol{x}}}
\newcommand{\Xc}{\mathcal{X}}
\newcommand{\Xcn}{\mathcal{X}^n}

\newcommand{\yb}{\boldsymbol{y}}
\newcommand{\Yb}{\boldsymbol{Y}}

\newcommand{\Yc}{\mathcal{Y}}
\newcommand{\Ycn}{\mathcal{Y}^n}

\newcommand{\Ac}{\mathcal{A}}

\newcommand{\Ccb}{\boldsymbol{\mathcal{C}}}
\newcommand{\Cc}{\mathcal{C}}
\newcommand{\Sc}{\mathcal{S}}

\newcommand{\Pc}{\mathcal{P}}
\newcommand{\Tc}{\mathcal{T}}

\newcommand{\Mc}{\mathcal{M}}

\newcommand{\Csf}{\mathsf{C}}
\newcommand{\Csfb}{\boldsymbol{\mathsf{C}}}

\newcommand{\RR}{\mathbb{R}}
\newcommand{\PP}{\mathbb{P}}
\newcommand{\EE}{\mathbb{E}}

\newcommand{\indicator}{\mathds{1}}

\DeclareMathOperator*{\argmax}{arg\,max}

\newtheorem{lemma}{Lemma}
\newtheorem{theorem}{Theorem}
\newtheorem{proposition}{Proposition}
\newtheorem{corollary}{Corollary}

\theoremstyle{definition}
\newtheorem{definition}{Definition}

\begin{document}
\title{Dual Domain Expurgated Error Exponents for Source Coding with Side Information}

\author{Mehdi Dabirnia, Hamdi Joudeh and Albert~Guill\'en i F\`abregas
\thanks{Mehdi Dabirnia is with Centre Tecnol\`ogic de Telecomunicacions de Catalunya (CTTC), 08860 Castelldefels, Spain; e-mail: {\tt mehdi.dabirnia@cttc.cat}.
Hamdi Joudeh is with the Department of Electrical Engineering, Eindhoven University of Technology, Eindhoven, 5600 MB, The Netherlands; e-mail: {\tt h.joudeh@tue.nl}.
Albert Guill\'en i F\`abregas is with the Department of Engineering, University of Cambridge, CB2 1PZ Cambridge, U.K., the Department of Signal Theory and Communications and the Institute of Mathematics (IMTech), Universitat Polit\`ecnica de Catalunya (UPC) 08034 Barcelona, Spain (\tt guillen@ieee.org).
}
\thanks{This work was supported in part by the European Research Council under Grants 101116550, 101142747 and 101158232, and in part by the Spanish Ministry of Economy and Competitiveness under Grants PID2020-116683GB-C22 and PID2021-128373OB-I00 and partially by the WAVE-XL project (PID2024-160457OB-I00) funded by MCIN/AEI/10.13039/501100011033 and by FEDER, UE.}
}

\date{\today}
\maketitle

\vspace{-1cm}
\begin{abstract}
We introduce an expurgation method for source coding with side information that enables direct dual-domain derivations of expurgated error exponents. Dual-domain methods yield optimization problems over few parameters, with any sub-optimal choice resulting in an achievable exponent, as opposed to primal-domain optimization over distributions. In addition, dual-domain methods naturally allow for general alphabets and/or memory. We derive two such expurgated error exponents for different random-coding ensembles in the case where the decoder is possibly mismatched with respect to the source and side information joint distribution. We show the better of the exponents coincides with the Csisz\'ar-K\"orner exponent obtained via a graph decomposition lemma. We show some numerical examples that illustrate the differences between the two exponents and show that in the case of source coding without side information, the expurgated exponent coincides with the error exponent of the source optimal code.
\end{abstract}


%

\newpage
\section{Introduction}
\label{sec:intro}

Consider a pair of discrete memoryless correlated sources with finite alphabets $\mathcal{X}$ and $\mathcal{Y}$, and joint distribution $P_{XY}$.
The pair of i.i.d. sequences of length $n$ generated at random by the two sources are denoted by 
$\Xb =  X_1, X_2,\ldots, X_n $
and 
$\Yb =  Y_1, Y_2,\ldots, Y_n $, while realizations are denoted with a lowercase font as $\xb =  x_1, x_2,\ldots, x_n $ and $\yb =  y_1, y_2,\ldots, y_n $.
An encoder observes $\Xb$ and maps it into a codeword (or bin) $m$, chosen from a set of $M$ codewords. A standard block source code $\Cc(n,R)$ is defined as a mapping $\phi:\Xcn\to\Mc$ of source sequences $\xb\in\Xcn$ to a set of codewords $\Mc=\{1,\dotsc,M\}$ where $R=\frac{\log M}{n}$ is the code rate. At the receiver end, the decoder observes the codeword $m$ and, along with the side information $\yb$, provides an estimate of the original source sequence. More formally, the decoder $\psi:\Mc\times\Ycn\to\Xcn$ maps each codeword and side information sequence back into a source sequence $\hatxb$. The decoder makes an error whenever $\psi(\phi(\xb),\yb)\neq \xb$ and the probability of error is given by
\begin{equation}
p_e = \PP[\psi(\phi(\Xb),\Yb)\neq \Xb].
\end{equation}
This setting is known as block 
source coding with decoder side information, or Slepian-Wolf coding \cite{Slepian1973,gallager1979source}. In this paper, we are interested in the exponential decay rate of the decoding error probability with block-length $n$, known as the error exponent or reliability function.

In \cite{gallager1979source}, Gallager derived an achievable error exponent for the above setting using random coding, or random binning, and maximum a posteriori probability (MAP) decoding. In Gallager’s
random coding ensemble, source sequences are mapped to codewords (or bins) uniformly at random in
a pairwise independent manner. The corresponding achievable exponent is
\begin{equation}
\label{eq:Er_dual}
    E_{\mathrm{r}}(R) = \max_{\rho \in [0,1]}  \rho R -  E_{0}(\rho)
\end{equation}
where the function $E_{0}(\rho)$ is defined as
\begin{equation}
\label{eq:E0_function}
    E_{0}(\rho) \triangleq \log\sum\limits_{y \in \mathcal{Y}} P_Y(y)\bigg( \sum\limits_{x \in \mathcal{X}} P_{X|Y}(x|y)^{\frac{1}{1+\rho}}  \bigg)^{1+\rho} .
\end{equation}
The function $E_{0}(\rho)$ is known to be related to Arimoto’s conditional Rényi entropy as $E_{0}(\rho) = \rho H_{\frac{1}{1+\rho}}(X|Y)$ \cite{arimoto1977information}.
We refer to $E_{\mathrm{r}}(R)$ as the \emph{random-coding error exponent}, due to close resemblance in proof and form to the random-coding error exponent in channel coding \cite{Gallager1965,gallager1968ita}. 
The form in \eqref{eq:Er_dual} is known as a \emph{dual domain} expression.
An equivalent \emph{primal domain}\footnote{Dual domain expressions can be derived from their primal domain counterparts using Lagrange duality techniques.} expression for $ E_{\mathrm{r}}(R)$ is given as
\begin{equation}
\label{eq:Er_primal}
    E_{\mathrm{r}}(R) = \min_{P_{\tilde{X}\tilde{Y}}} D(P_{\tilde{X}\tilde{Y}} \| P_{XY}) + \big| R - H(\tilde{X}|\tilde{Y}) \big|^{+} 
\end{equation}
where the minimization is over all joint distributions $P_{\tilde{X}\tilde{Y}}$ on $\mathcal{X} \times \mathcal{Y}$. 
A direct achievability proof that yields the primal domain expression in \eqref{eq:Er_primal} is given by Csisz\'{a}r and K\"{o}rner in \cite{Csiszar1980}, where type analysis and universal minimum entropy (ME) decoding are employed.

While $ E_{\mathrm{r}}(R)$ provides a lower bound on the reliability function, Gallager \cite{gallager1979source} also derived a corresponding upper bound, which takes the following form
\begin{equation}
\label{eq:Esp_dual}
    E_{\mathrm{sp}}(R) = \sup_{\rho \geq 0}  \rho R - E_{0}(\rho).
\end{equation}
We will refer to this as the \emph{sphere-packing exponent}, as it closely resembles the sphere-packing exponent in channel coding \cite{Shannon1967}.
The sphere-packing exponent $E_{\mathrm{sp}}(R) $ is derived by providing the side information sequence $\Yb$ to the encoder, and is in fact equal to the error exponent of conditional source coding, where side information is available to both encoder and decoder.
The sphere packing exponent admits an equivalent primal domain expression \cite{Csiszar1980} given as
\begin{equation}
\label{eq:Esp_primal}
    E_{\mathrm{sp}}(R) =\min_{P_{\tilde{X}\tilde{Y}} : H(\tilde{X}|\tilde{Y}) \geq R } D(P_{\tilde{X}\tilde{Y}} \| P_{XY}).
\end{equation}

Both $E_{\mathrm{r}}(R)$ and $E_{\mathrm{sp}}(R)$ are non-negative for rates in the range $H(X|Y)  < R < \log |\mathcal{S}(X)|$, where $\mathcal{S}(X) \subseteq \mathcal{X}$ is the support of $P_X$; and are equal to zero for $0 \leq R \leq H(X|Y)$.
Moreover, $E_{\mathrm{r}}(R)$ and $E_{\mathrm{sp}}(R)$ coincide in the range $H(X|Y) < R \leq R_{\mathrm{cr}}$, where $R_{\mathrm{cr}}$ is the largest rate at which the convex curve $E_{\mathrm{sp}}(R)$ meets its supporting line of slope 1. Note that $R_{\mathrm{cr}}$ is reminiscent of the \emph{critical rate} in channel coding.
\subsection{Expurgated exponent}
For a range of rates above the critical rate $R_{\mathrm{cr}}$, a tighter achievable error exponent was derived by Csisz{\'a}r and K{\"o}rner in \cite{Csi81}, and is given by
\begin{align}
\label{eq:Eex_primal}
E_{\rm{ex}}(R)=\min_{P_{\tilde{X}}} \bigg \lbrace D(P_{\tilde{X}}\|P_X)  + \min_{P_{\hat{X}\tilde{X}}:P_{\hat{X}}=P_{\tilde{X}}, H(\hat{X}|\tilde{X})\geq R}
\Big\lbrace\mathbb{E}\big[d(\hat{X},\tilde{X})\big]  + R  -  H(\hat{X}|\tilde{X})\Big\rbrace\bigg\rbrace
\end{align}
where $d(\hat{x},\tilde{x})$ is the Bhattacharyya distance between $P_{Y|X}(\cdot|\hat{x})$ and $P_{Y|X}(\cdot|\tilde{x})$, defined as 
\begin{equation}
 d(\hat{x},\tilde{x}) \triangleq - \log \sum\limits_{y \in \mathcal{Y}} \sqrt{ P_{Y|X}(y|\hat{x}) P_{Y|X}(y|\tilde{x}) }.  
\end{equation}
The exponent $E_{\rm{ex}}(R)$ is often referred to as the \emph{expurgated exponent}, since it can be seen as a source coding counterpart to the expurgated exponent in channel coding when expressed in the primal domain \cite{Csi81}, \cite[Problem 10.18]{Csiszar2011}. 
Relating the Slepian-Wolf source coding problem to a counterpart channel coding problem with input $\Xb$ and output $\Yb$ is in fact the first step in the proof of Csisz{\'a}r and K{\"o}rner \cite{Csi81}, as well as a later proof by Ahlswede and Dueck \cite{Ahlswede1982}.

In deriving the expurgated exponent $E_{\rm{ex}}(R)$, Csisz{\'a}r and K{\"o}rner \cite{Csi81} employ a type-by-type block coding scheme, in which source sequences assigned to the same codeword (or bin) all have the same type. 
Since source sequences of the same type have the same probability, optimal MAP decoding reduces to maximum likelihood (ML) decoding in this case, which was used to derive \eqref{eq:Eex_primal}.
In the code construction, instead of relying on random coding, a graph-theoretic decomposition lemma is used to show that every type class can be partitioned into so-called ``balanced'' sets with a favorable packing property, and these balanced sets are taken as bins.
The same exponent was later derived by Ahlswede and Dueck in \cite{Ahlswede1982} by exploiting the connection to the counterpart channel coding problem and using permutation codes. 
In particular, their approach relies on covering each source type class using permutations of a good constant-composition channel code of the same type, that achieves the channel coding expurgated exponent. 

Both the Csisz{\'a}r-K{\"o}rner and Ahlswede-Dueck proofs of the expurgated exponent rely heavily on type analysis and combinatorial arguments, and do not use random coding or expurgation, at least not in a direct manner. Moreover, their derivations yield the primal expression in \eqref{eq:Eex_primal}, and currently it is not known whether this expression admits an equivalent dual form. 

The main objective for our present work is to find a dual-domain derivation for the expurgated exponent in source coding with side information, mirroring Gallager's original proof of the expurgated exponent in channel coding \cite{Gallager1965}. Gallager's expurgation approach does not directly apply here as in the source coding setting, one cannot simply remove a ``bad'' fraction of source sequences. A dual-domain derivation is also expected to yield non-asymptotic bounds that are valid for settings with arbitrary memory, countable source alphabets, and general side information alphabet. 
Furthermore, compared to the primal expression in \eqref{eq:Eex_primal}, a dual domain expression will likely have far fewer parameters and will hence be easier to compute and evaluate. 
\subsection{Mismatched Decoding}
Another motivation for our current work is to evaluate error exponents for source coding with side information under generic, and possibly mismatched, decoding metrics. Mismatched decoding, thoroughly studied in the context of channel coding, is the setting for which the decoder employs a fixed decoding metric, not necessarily related to the probability law describing the system (see \cite{scarlett2020information} and references therein). Mismatched decoding naturally arises in cases where the decoder cannot accurately estimate the system's parameters, or in cases where, for complexity reasons, the one prefers an alternative decoding metric.
To this end, in the considered setting we adopt a maximum metric decoder of the form
\begin{equation}
	\psi(m,\yb)=\argmax_{\xb\in\Xcn:\phi(\xb)=m}q(\xb,\yb),
	\label{eq:dec}
\end{equation}
where $q$ is an arbitrary non-negative decoding metric. 
Upon observing $m$ and $\yb$, the decoder chooses the source sequence $\xb$ whose metric is highest out of those encoded to $m$. 
The decoding metric is said to be memoryless if $q(\xb,\yb)=\prod_{i=1}^{n}q(x_i,y_i)$.
In this case, and under memoryless source and side information, the optimal MAP decoder is recovered by choosing the metric as $q(x,y) = P_{X|Y}(x|y)$. Otherwise, the decoder is said to be mismatched.
As we shall see further on, the majority of our derivations are valid for an arbitrary metric $q(\xb,\yb)$.
With that said, in our single-letter asymptotic results, we choose to focus on the memoryless case.

Under an arbitrary memoryless decoding metric $q$, primal-domain random coding and expurgated exponents that generalize \eqref{eq:Er_primal} and \eqref{eq:Eex_primal} can be distilled from the results of Csisz\'ar and K\"{o}rner \cite{Csi81}, obtained using their graph decomposition approach. These results are derived for a generic class of decoders called $\alpha$-decoders, which includes the memoryless mismatched decoder as a special case. The corresponding random coding and expurgated exponents are given by 
\begin{equation}
	E^{\mathrm{ck}}_{q,\mathrm{r}}(R)= \min_{P_{\hat{X}\tilde{X}\tilde{Y}}\in\Tc} D(P_{\tilde{X}\tilde{Y}}\|P_{XY})+|R-H(\hat{X}|\tilde{Y})|^{+},
	\label{eq:RC_CK_SW_exponent}
\end{equation}
and
\begin{equation}
	E^{\mathrm{ck}}_{q,\mathrm{ex}}(R)=\min_{\substack{P_{\hat{X}\tilde{X}\tilde{Y}}\in\mathcal{T}\\H(\hat{X}|\tilde{X})\geq R}}D(P_{\tilde{X}\tilde{Y}}\|P_{XY})+R-H(\hat{X}|\tilde{X},\tilde{Y})
	\label{eq:EX_CK_SW_exponent}
\end{equation}
respectively, where the set of distributions $\mathcal{T}$ is defined as 
\begin{equation}
\mathcal{T}=\Big\lbrace P_{\hat{X}\tilde{X}\tilde{Y}}\in\Pc(\Xc\times\Xc\times\Yc):P_{\hat{X}}=P_{\tilde{X}},\mathbb{E}\big[\log q(\hat{X},\tilde{Y})\big]\geq\mathbb{E}\big[\log q(\tilde{X},\tilde{Y})\big]\Big\rbrace.
\end{equation}
These exponents reduce to their matched counterparts in \eqref{eq:Er_primal} and \eqref{eq:Eex_primal} when the decoding metric $q(x,y)=P_{Y|X}(y|x)$ is chosen to be the ML decoder, which is optimal in this case due to the type-by-type coding scheme used in \cite{Csi81}. 

From \eqref{eq:RC_CK_SW_exponent} and \eqref{eq:EX_CK_SW_exponent}, we obtain the Csisz\'ar-K\"orner's achievable exponent given by
\begin{equation}
	E^{\mathrm{ck}}_q(R)=\max\left\{E^{\mathrm{ck}}_{q,\mathrm{r}}(R),E^{\mathrm{ck}}_{q,\mathrm{ex}}(R)\right\}.
	\label{eq:ck}
\end{equation}
Dual-domain derivations and expressions for \eqref{eq:RC_CK_SW_exponent} and \eqref{eq:EX_CK_SW_exponent}, and hence \eqref{eq:ck}, are currently not known. This is a gap in the literature that we aim to address in the present paper.

There is a body of work on variable-rate codes for source coding with side information, see e.g. \cite{Weinberger2015,Chen2017,Tamir2021}. In the variable-rate setting, the sequences of the types that are bounded away from source distribution $P_{X}$ (assuming that the distribution $P_{X}$ is known exactly) are losslessly encoded. This relaxation of the rate constraint does not affect the average rate asymptotically and improves the error exponent when the dominant error of the fixed-rate codes are caused by one of those types. The works \cite{Weinberger2015,Tamir2021} analyze the trade-off between the error-exponent and excess-rate exponent for the variable-rate setting. While the idea of using variable-rate codes for the source coding with side-information problem is well-investigated for the matched case, it does not immediately extend to the mismatched case. In this paper we focus on the setting of fixed-length (fixed-rate) codes and do not consider variable-rate codes.

\subsection{Contribution}

The main contributions of the paper are as follows:
\begin{itemize}
	\item In Section \ref{sec:main results}, we present our main theorems on dual-domain achievable exponents with mismatched decoding for two random coding ensembles, namely standard ensemble and type-by-type ensemble. These theorems show the existence of a code that attains the maximum of two exponents, namely random coding and expurgated exponents. The dual domain exponents for type-by-type ensemble and expurgated exponent for standard ensemble appear for the first time in the literature. We further show the equivalence of the dual domain exponents to the Csisz\'ar-K\"orner's achievable exponent in \eqref{eq:ck}. We also present relative comparisons of these achievable exponents and state the family of metrics that attain the corresponding optimal exponent. The achievable rates for these ensembles together with their relation to the generalized mutual information and LM rates of mismatched decoding in channel coding are also presented.
	
	\item In Section \ref{sec:exp}, we introduce an expurgation method for source coding that is valid for general source and side information models and arbitrary decoding metrics. 
	The expurgation method works for either of the two code ensembles and can be summarized as follows. We first use Gallager's expurgation technique, developed for channel coding \cite{Gallager1965}, to show that there exists a code in the ensemble such that at least half of the source sequences satisfy a desired error bound. We expurgate the ``bad'' half of source sequences, encode them separately into a new set of codewords, and apply expurgation again. The error bound derived in the previous iteration remains valid here, since we now have fewer source sequences. The procedure stops after (at most) $k = n \log_2 |\mathcal{X}|$ iterations, once all source sequences are exhausted. Combining the expurgated codes from all iterations, we obtain a code in which all the source sequences satisfy the desired error bound.

	\item Based on our expurgation method, in Sections \ref{sec:proof_th1} and \ref{sec:proof_th2}, we derive non-asymptotic bounds that are valid for any discrete source model with arbitrary side information alphabet and arbitrary decoding metric.\end{itemize}

\subsection{Notation}

We use bold symbols for vectors (e.g. $\xb$), and denote the corresponding $i$-th entry using a subscript (e.g. $x_i$). The set of all probability distributions on an alphabet $\Xc$, is denoted by $\Pc(\Xc)$, and the set of all empirical distributions on a vector in $\Xcn$ is denoted by $\Pc_n(\Xc)$. For a given type $\hat{P}\in\Pc_n(\Xc)$, the type
class $\Tc_n(\hat{P})$ is defined to be the set of all sequences in $\Xcn$ with type $\hat{P}$.
The probability of an event is denoted by $\mathbb{P}[\cdot]$. The marginals of a joint
distribution $P_{XY}(x,y)$ are denoted by $P_{X}(x)$ and $P_{Y}(y)$. We write $P_X=Q_X$ to denote element-wise equality between two probability distributions on the same alphabet. Expectation with respect to a joint distribution $P_{XY}(x,y)$ is denoted by $\mathbb{E}_{P}[\cdot]$, or simply $\mathbb{E}[\cdot]$ when the associated probability distribution is understood from the context. Given a distribution $Q_{X}$ and conditional distribution $W_{Y|X}$, we write $Q_{X}\times W_{Y|X}$ to denote the resulting joint distribution and the corresponding mutual information is written as $I(Q_{X},W_{Y|X})$, or simply $I(X;Y)$ when the distribution is understood from the context. We use standard notation for the entropy $H(X)$ (also sometimes shown as $H(Q)$ to emphasize the distribution of the random variable), conditional entropy $H(X|Y)$ and divergence $D(P\|Q)$. All logarithms have base $e$ and all rates are in units of nats including in the example. We denote the indicator function of an event by $\mathds{1}[\cdot]$.

\section{Main Results}
\label{sec:main results}

In this section, we introduce our main results, which are dual domain expressions of expurgated error exponents. These expressions are obtained directly in the dual domain by means of an expurgation method detailed in Section \ref{sec:exp} for two different code ensembles. These code ensembles differ on whether the set of source sequences and codewords are partitioned or not. The relative merits of these ensembles are illustrated by means of a numerical example. Proofs of the results are given in Sections \ref{sec:proof_th1}, \ref{sec:proof_th2} and \ref{sec:proof_primal_dual}, respectively.

\subsection{Standard Block Random Coding}

\begin{definition}
	The \emph{Standard Random Coding Ensemble} is the set of all ($n,R$) standard block codes for source alphabet $\Xc$ with a probability measure over the codes having the following property: each source sequence is assigned independently and with equal probability $\frac{1}{M}$ into one of the $M={\rm e}^{nR}$ codewords.
\end{definition}
For the standard random coding ensemble we have the following result concerning achievable error exponents for memoryless sources employing a memoryless decoding metric.

\begin{theorem}
\label{Thm:Exponent_St}
	For every $R>0$ and every distribution $P_{XY}\in\Pc(\Xc\times\Yc)$ there exists a standard block source code with maximum metric decoder \eqref{eq:dec} employing decoding metric $q(x,y)$ that achieves the exponent
	\begin{equation}
		E_{q}(R)=\max\{E_{q,\mathrm{r}}(R),E_{q,\mathrm{ex}}(R)\}
		\label{eqn:exp_St}
	\end{equation}
	
	where
	\begin{equation}
	E_{q,\mathrm{r}}(R) = \sup_{0\leq\rho\leq 1,s\geq 0}\rho R-\log\sum\limits_{x\in\Xc,y\in\Yc}P_{XY}(x,y)\left(\sum\limits_{\bar{x}\in\Xc}\left(\frac{q(\bar{x},y)}{q(x,y)}\right)^s\right)^{\rho},\label{eqn:exp_St_RC}
\end{equation}
and
\begin{equation}
	E_{q,\mathrm{ex}}(R) = \sup_{\rho\geq 1,s\geq 0}\rho R-\log\sum\limits_{x\in\Xc}\left(\sum\limits_{\bar{x}\in\Xc}\left(\sum\limits_{y\in\Yc}P_{XY}(x,y)\left(\frac{q(\bar{x},y)}{q(x,y)}\right)^s\right)^{\frac{1}{\rho}}\right)^{\rho}.\label{eqn:exp_St_EX}
\end{equation}
and the optimization is over real parameters $\rho$ and $s$.
	\label{th:standard}
\end{theorem}

\begin{proof}
The proof is structured in two parts. The first part consists of showing that there exists a code for which the error probability for every source sequence $\xb\in\Xc^n$ meets a certain upper bound (stated in Lemma \ref{Lem:expurgation_St}). This bound is derived in Section \ref{subsec:standard} by introducing a method for expurgation in source coding with standard ensemble. The second part of the proof is the analysis of the bound in Lemma \ref{Lem:expurgation_St}, which results in the exponent \eqref{eqn:exp_St}. This can be found in Section \ref{sec:proof_th1}.
\end{proof}

In deriving the error exponent of Theorem \ref{th:standard} we derive $n$-letter bounds in \eqref{eqn:bound_expurgated} and \eqref{eqn:random_coding_bound} that are valid for any discrete source with arbitrary side information alphabets and arbitrary decoding metrics without the memoryless assumption. These bounds can be used to derive error exponents for general source models.

\subsection{Type-by-Type Random Coding}

We now introduce a different block random coding ensemble that encodes each source type separately. As will be shown next, this attains potentially a higher error exponent than the standard block code ensemble under the same decoding metric. Consider partitioning the codeword set $\Mc=\{1,\cdots,M\}$ into $|\Pc_n(\Xc)|$ subsets as: $\Mc=\bigcup\limits_{i=1}^{|\Pc_n(\Xc)|}\Mc_i$ where $\Mc_i$ is the codeword set for source type $\hat{P}_i$, $i\in\{1,\cdots,|\Pc_n(\Xc)|\}$, $|\Mc_i|=\frac{M}{|\Pc_n(\Xc)|}$ and $\Mc_i\cap\Mc_j=\emptyset$ for $i\neq j$. 

A type-by-type block source code $\Ccb=\{\Cc_1,\cdots,\Cc_{|\Pc_n(\Xc)|}\}$ is the union of $|\Pc_n(\Xc)|$ codes where each $\Cc_i$ is an $(n,R_i)$ block code for source type $\hat{P}_i$ with mapping function $\phi_{i}:\Tc_n(\hat{P}_i)\to\Mc_i$ and rate $R_i=\frac{\log M}{n}-\delta_n$ for $i\in\{1,\cdots,|\Pc_n(\Xc)|\}$, where $\delta_n=\frac{\log|\Pc_n(\Xc)|}{n}$. In other words, code $\Cc_i$ is a code for the source sequences in source type class $\Tc_n(\hat{P}_i)$ with codeword set $\Mc_i$ . 
Observe that by construction, every code $\Cc_i$ has the same rate, and that an error can only occur between source sequences of the same type. We also note that the effect of partitioning the codeword set on the coding rate vanishes asymptotically since $\delta_n\to 0$ as $n\to\infty$.

\begin{definition}\label{def:type_based_RC}
	The \emph{Type-by-Type Random Coding Ensemble} is the set of all ($n,R$) block codes for the source alphabet $\Xc$ with a probability measure over the codes having the following property: For every source type $\hat{P}_i$, $i\in\{1,\cdots,|\Pc_n(\Xc)|\}$, each source sequence $\xb\in\Tc_n(\hat{P}_i)$ is independently assigned with equal probability $\frac{1}{|\Mc_i|}$ to each of the codewords in $\Mc_i$.
\end{definition}

The decoder for a type-by-type code $\Ccb$ is a set of mappings $\psi_{i}:\Mc_i\times\Ycn\to\Tc_n(\hat{P}_i)$ for every $i\in\{1,\cdots,|\Pc_n(\Xc)|\}$. Similarly to the standard ensemble, we consider using a maximum metric decoder as follows. For every $m\in\Mc_i$
\begin{equation}
	\psi_{i}(m,\yb)=\argmax_{\xb\in\Tc_n(\hat{P}_i):\phi_i(\xb)=m}q(\xb,\yb),\label{eq:dec_tbt}
\end{equation}
where $q(\xb,\yb)$ is an arbitrary non-negative decoding metric. For the type-by-type random coding ensemble we have the following result concerning the achievable error exponents for memoryless sources employing a memoryless decoding metric.

\begin{theorem}[Type-by-Type Random Coding]
\label{th:tbt}
For every $R>0$ and every distribution $P_{XY}\in\Pc(\Xc\times\Yc)$ there exists a type-by-type block source code with maximum metric decoder \eqref{eq:dec_tbt} employing decoding metric $q(x,y)$ that achieves the exponent
	\begin{equation}
		E_{q}^{\rm tt}(R)=\max\{E^{\rm tt}_{q,\mathrm{r}}(R),E^{\rm tt}_{q,\mathrm{ex}}(R)\}\label{eq:Exponent_Type}
	\end{equation}
	\label{Thm:Exponent_Type}
where
\begin{equation}
	E^{\rm tt}_{q,\mathrm{r}}(R) = \sup_{0\leq\rho\leq 1,s\geq 0,a(\cdot)}\rho R-\log\sum\limits_{x\in\Xc,y\in\Yc}P_{XY}(x,y)\left(\sum\limits_{\bar{x}\in\Xc}\frac{e^{a(\bar{x})}}{e^{a(x)}}\left(\frac{q(\bar{x},y)}{q(x,y)}\right)^s\right)^{\rho},\label{eqn:exp_Type_RC}
\end{equation}
\begin{equation}
	E^{\rm tt}_{q,\mathrm{ex}}(R) = \sup_{\rho\geq 1,s\geq 0,a(\cdot)}\rho R-\log\sum\limits_{x\in\Xc}\left(\sum\limits_{\bar{x}\in\Xc}\left(\sum\limits_{y\in\Yc}P_{XY}(x,y)\frac{e^{a(\bar{x})}}{e^{a(x)}}\left(\frac{q(\bar{x},y)}{q(x,y)}\right)^s\right)^{\frac{1}{\rho}}\right)^{\rho}\label{eqn:exp_Type_EX}
\end{equation}
and the optimization is over real parameters $\rho,s$ and real-valued functions $a:\Xc\to\RR$.

\end{theorem}

\begin{proof}
	The proof is structured in two parts. The first part consists of showing that there exists a code for which the error probability for every source sequence $\xb\in\Xc^n$ meets a certain upper bound (stated in Lemma \ref{Lem:expurgation_Type}). This bound is derived in Section \ref{subsec:type-by-type} by extending the introduced expurgation method to type-by-type ensemble. The second part of the proof is the analysis of the bound in Lemma \ref{Lem:expurgation_Type}, which results in the exponent \eqref{eq:Exponent_Type}. This can be found in Section \ref{sec:proof_th2}.
\end{proof}

In deriving the error exponent of Theorem \ref{th:tbt} we derive $n$-letter bounds in \eqref{eqn:bound_expurgated_Type} and \eqref{eqn:random_coding_bound_Type} that are valid for any discrete source with arbitrary side information alphabets and arbitrary decoding metrics without the memoryless assumption. These bounds can be used to derive error exponents for general source models.

The following result shows that the error exponent introduced in Theorem \ref{th:tbt} coincides with Csisz\'ar and K\"orner's \cite{Csi81} for the case of using memoryless metric.

\begin{proposition}[Primal-dual equivalence]
	The dual-domain error exponent derived in Theorem \ref{Thm:Exponent_Type} coincides with the Csisz\'ar-K\"orner exponent \eqref{eq:ck} derived in the primal domain via graph decomposition, i.e.,
	\begin{equation}
		E^{\rm tt}_{q}(R) = E^{\mathrm{ck}}_q(R).
	\end{equation}
	\label{pr:equaiv}
\end{proposition}

\begin{proof}
	The proof has two parts. In the first part the equivalence of $E^{\mathrm{ck}}_{q,\mathrm{r}}(R)$ in \eqref{eq:RC_CK_SW_exponent} to $E^{\rm tt}_{q,\mathrm{r}}(R)$ in \eqref{eqn:exp_Type_RC} is shown. In the second part the equivalence of $E^{\mathrm{ck}}_{q,\mathrm{ex}}(R)$ in \eqref{eq:EX_CK_SW_exponent} to $E^{\rm tt}_{q,\mathrm{ex}}(R)$ in \eqref{eqn:exp_Type_EX} is shown. The proof can be found in Section \ref{sec:proof_primal_dual}.
\end{proof}

In proving the equivalence of the primal and dual forms of the type-by-type exponent we show the following relations between type-by-type exponents and the exponent of constant composition codes for a corresponding channel as follows,
\begin{align}
	E^{\mathrm{tt}}_{q,\rm r}(R)&=D(Q\|P_{X})+ E_{q,\mathrm{r}}^{\mathrm{cc}}\left(H(Q)-R,Q,P_{Y|X}\right),\\
	E^{\mathrm{tt}}_{q,\rm ex}(R)&=D(\tilde{Q}\|P_{X})+ E_{q,\mathrm{ex}}^{\mathrm{cc}}\left(H(\tilde{Q})-R,\tilde{Q},P_{Y|X}\right),
\end{align}
where $E_{q,\mathrm{r}}^{\mathrm{cc}}\left(\cdot,\cdot,\cdot\right)$ and $E_{q,\mathrm{ex}}^{\mathrm{cc}}\left(\cdot,\cdot,\cdot\right)$ are defined in \eqref{eqn:Channel_CC_rc_exponent} and \eqref{eqn:Channel_CC_expurgated_exponent}, respectively, and
\begin{align}
	Q(x)&=CP_{X}(x)\sum\limits_{y\in\Yc}P_{Y|X}(y|x)\left(\sum\limits_{\bar{x}\in\Xc}\frac{e^{a^*(\bar{x})}}{e^{a^*(x)}}\left(\frac{q(\bar{x},y)}{q(x,y)}\right)^{s^*}\right)^{\rho^*},\label{eq:Q}\\
	\tilde{Q}(x)&=\tilde{C}P_{X}(x)\left(\sum\limits_{\bar{x}\in\Xc}\left(\sum\limits_{y\in\Yc}P_{Y|X}(y|x)\frac{e^{a^*(\bar{x})}}{e^{a^*(x)}}\left(\frac{q(\bar{x},y)}{q(x,y)}\right)^{s^*}\right)^{\frac{1}{\rho^*}}\right)^{\rho^*}\label{eq:tildeQ},
\end{align}
and $C$ and $\tilde{C}$ are the normalization constants and $s^*, \rho^*$ and $a^*(\cdot)$ in \eqref{eq:Q} and \eqref{eq:tildeQ} are the optimizing choices, respectively in \eqref{eqn:exp_Type_RC} and \eqref{eqn:exp_Type_EX} for rate $R$. 

The following results are relative comparisons among the exponents introduced in the previous theorems, stating the families of metrics that attain the optimal exponent.
\begin{corollary}
	For every mismatched decoder employing decoding metric $q(x,y)$, we have that
	\begin{equation}
		E_{q,\mathrm{r}}(R) \leq E^{\rm tt}_{q,\mathrm{r}}(R) \leq E_{\mathrm{r}}(R),\label{eqn:RC_exponent_inequalities}
	\end{equation}
	where $E_{\mathrm{r}}(R)$ is defined in \eqref{eq:Er_dual} and the equality $E^{\rm tt}_{q,\mathrm{r}}(R)=E_{\mathrm{r}}(R)$ holds for any
	\begin{equation}
		q(x,y)=e^{b(x)+c(y)}P_{X|Y}(x|y)^{\tau}\label{eqn:RC_Type_optimal_decoders}
	\end{equation}
	with arbitrary $b(x)$, $c(y)$ and $\tau>0$. Furthermore, the equality $E_{q,\mathrm{r}}(R)=E_{\mathrm{r}}(R)$ holds for any
	\begin{equation}
		q(x,y)=e^{c(y)}P_{X|Y}(x|y)^{\tau}\label{eqn:RC_St_optimal_decoders}
	\end{equation}
	with arbitrary $c(y)$ and $\tau>0$.
		
	\label{cor:RC_exponent_comparison}
\end{corollary}
	
\begin{proof}
	The first inequality in \eqref{eqn:RC_exponent_inequalities} follows by setting $a(x)$ equal to a constant for all $x$ in \eqref{eqn:exp_Type_RC}. To show the second inequality we write the summation inside the logarithm in \eqref{eqn:exp_Type_RC} as
	\begin{equation}
		\sum\limits_{y\in\Yc}P_{Y}(y)\bigg(\sum\limits_{x\in\Xc}P_{X|Y}(x|y)e^{-\rho a(x)}q(x,y)^{-s\rho}\bigg)\bigg(\sum\limits_{\bar{x}\in\Xc}e^{a(\bar{x})}q(\bar{x},y)^s\bigg)^{\rho}.\label{eqn:Type_RC_log_argument}
	\end{equation}
	We now apply H\"older's inequality:
	\begin{equation}
		\bigg(\sum\limits_{i}a_i^{\frac{1}{1+\rho}}b_i^{\frac{1}{1+\rho}}\bigg)^{1+\rho}\leq\bigg(\sum\limits_{i}a_i\bigg)\bigg(\sum\limits_{i}b_i^{\frac{1}{\rho}}\bigg)^{\rho}.\label{eqn:Holder inequality}
	\end{equation}
	Identifying $a_i=P_{X|Y}(x|y)e^{-\rho a(x)}q(x,y)^{-s\rho}$ and $b_i=(e^{a(x)}q(x,y)^s)^{\rho}$, we conclude that \eqref{eqn:Type_RC_log_argument} is lower bounded by
	\begin{equation}
		\sum\limits_{y\in\Yc}P_{Y}(y)\bigg(\sum\limits_{x\in\Xc}P_{X|Y}(x|y)^{\frac{1}{1+\rho}}\bigg)^{1+\rho}.
	\end{equation}
	The necessary and sufficient condition for equality in \eqref{eqn:Holder inequality} is that $a_i=c b_i^{\frac{1}{\rho}}$ for all $i$ and some positive constant $c$. Applying this to \eqref{eqn:Type_RC_log_argument} for all $y$ in outer summation, we show that $E^{\rm tt}_{q,\mathrm{r}}(R)=E_{\mathrm{r}}(R)$ holds for any decoding metric of the form given in \eqref{eqn:RC_Type_optimal_decoders}. Similarly, rewriting the summation inside logarithm in \eqref{eqn:exp_St_RC} and applying H\"older's inequality we can show that $E_{q,\mathrm{r}}(R)=E_{\mathrm{r}}(R)$ holds if and only if the decoding metric is of the form given in \eqref{eqn:RC_St_optimal_decoders}.
\end{proof}
	
Corollary \ref{cor:RC_exponent_comparison} shows that both type-by-type and standard ensembles recover $E_{\mathrm{r}}(R)$ with a family of metrics given by \eqref{eqn:RC_Type_optimal_decoders} and \eqref{eqn:RC_St_optimal_decoders}, where the latter is a subset of the former and both include the MAP decoding as a special case.
	
In order to compare the mismatched error exponents introduced in previous theorems with their matched counterparts, we define the following matched expurgated exponents
\begin{equation}
	E_{\mathrm{ex}}(R) = \sup_{\rho\geq 1,s\geq 0}\rho R-\log\sum\limits_{x\in\Xc}\left(\sum\limits_{\bar{x}\in\Xc}\left(\sum\limits_{y\in\Yc}P_{XY}(x,y)\left(\frac{P_{X|Y}(\bar{x}|y)}{P_{X|Y}(x|y)}\right)^s\right)^{\frac{1}{\rho}}\right)^{\rho}\label{eqn:Ex_st_MAP},
\end{equation}
\begin{equation}
	E_{\mathrm{ex}}^{\rm tt}(R) = \sup_{\rho\geq 1,a(\cdot)}\rho R-\log\sum\limits_{x\in\Xc}P_X(x)\left(\sum\limits_{\bar{x}\in\Xc}\left(\sum\limits_{y\in\Yc}\frac{e^{a(\bar{x})}}{e^{a(x)}}\sqrt{P_{Y|X}(y|x)P_{Y|X}(y|\bar{x})}\right)^{\frac{1}{\rho}}\right)^{\rho}.\label{eqn:Ex_type_MAP}
\end{equation}

Similarly to Corollary \ref{cor:RC_exponent_comparison}, we have the following.
	\begin{corollary}
	For every mismatched decoder employing decoding metric $q(x,y)$, we have that
	\begin{align}
		E_{q,\mathrm{ex}}(R) &\leq E^{\rm tt}_{q,\mathrm{ex}}(R) \leq E^{\rm tt}_{\mathrm{ex}}(R),\label{eqn:EX_exponent_inequalities1}\\
		E_{\mathrm{ex}}(R) &\leq E^{\rm tt}_{\mathrm{ex}}(R),\label{eqn:EX_exponent_inequalities2}
	\end{align}
	and the equality $E^{\rm tt}_{q,\mathrm{ex}}(R)=E^{\rm tt}_{\mathrm{ex}}(R)$ for all $0\leq R\leq \log|\Xc|$ holds for any
	\begin{equation}
		q(x,y)=e^{b(x)+c(y)}P_{X|Y}(x|y)^{\tau}\label{eqn:Ex_Type_optimal_decoders}
	\end{equation}
	with arbitrary $b(x)$, $c(y)$ and $\tau>0$, furthermore, the equality $E_{q,\mathrm{ex}}(R)=E^{\rm tt}_{\mathrm{ex}}(R)$ holds for a given $0\leq R\leq \log|\Xc|$ for any
	\begin{equation}
		q(x,y)=e^{c(y)}\left(e^{a(x)}\sqrt{P_{Y|X}(y|x)}\right)^{\tau}\label{eqn:Ex_St_optimal_decoders}
	\end{equation}
	where $a(\cdot)$ is the optimal choice for given $R$ in \eqref{eqn:Ex_type_MAP} and choice of $c(y)$ and $\tau>0$ are arbitrary.
	\label{cor:EX_exponent_comparison}
\end{corollary}
\begin{proof}
		The first inequality in \eqref{eqn:EX_exponent_inequalities1} follows by setting $a(x)$ equal to a constant for all $x$ in \eqref{eqn:exp_Type_EX}. The second inequality in \eqref{eqn:EX_exponent_inequalities1} follows from \cite[Lemma 4]{Csi81} and Proposition \ref{pr:equaiv}.
		
		For a metric of the form given in \eqref{eqn:Ex_Type_optimal_decoders}, after some simplification we can rewrite \eqref{eqn:exp_Type_EX} as
		\begin{equation}
			E^{\rm tt}_{q,\mathrm{ex}}(R) = \sup_{\rho\geq 1,s'\geq 0,r(\cdot)}\rho R-\log\sum\limits_{x\in\Xc}\left(\sum\limits_{\bar{x}\in\Xc}\left(\sum\limits_{y\in\Yc}P_{X,Y}(x,y)\frac{e^{r(\bar{x})}}{e^{r(x)}}\left(\frac{P_{Y|X}(y|\bar{x})}{P_{Y|X}(y|x)}\right)^{s'}\right)^{\frac{1}{\rho}}\right)^{\rho},\label{eqn:Ex_type_MAP_equiv}
		\end{equation}
		where $s'=\tau s$ and $r(x)=\rho a(x)+sb(x)+\tau s\log P_{X}(x)$. As can be found in Section \ref{sec:proof_primal_dual}, we can rewrite this in similar form to \eqref{eqn:Expurgated_exp_primal} as 
		\begin{equation}
			E^{\rm tt}_{q,\mathrm{ex}}(R) = \min_{Q}\left[ D(Q\|P_{X})+\sup_{\rho\geq 1}\left[E_{x}^{\rm cc}(Q,\rho)-\rho(H(Q)-R)\right]\right],
		\end{equation}
		where 
		\begin{align}
			E_{\rm{x}}^{\rm cc}(Q,\rho)=\sup_{s'\geq 0,r(\cdot)}-\rho\sum\limits_{x\in\Xc}Q(x)\log\sum\limits_{\bar{x}\in\Xc}Q(\bar{x})\left(\sum\limits_{y\in\Yc}P_{Y|X}(y|x)\frac{e^{r(\bar{x})}}{e^{r(x)}}\left(\frac{P_{Y|X}(y|\bar{x})}{P_{Y|X}(y|x)}\right)^{s'}\right)^{\frac{1}{\rho}}.\label{eqn:Ex_CC} 
		\end{align}
		It has been shown in \cite{scarlett2014} that $s'=\frac{1}{2}$ optimizes \eqref{eqn:Ex_CC} while the optimal choice of $r(\cdot)$ is unclear in general. Replacing the latter optimal choice in \eqref{eqn:Ex_type_MAP_equiv} and further simplification results in \eqref{eqn:Ex_type_MAP}.
		
		Replacing the metric $q(x,y)$ in \eqref{eqn:exp_St_EX} with \eqref{eqn:Ex_St_optimal_decoders} and setting $s=\frac{1}{\tau}$ results in \eqref{eqn:Ex_type_MAP}.
		
		The inequality in \eqref{eqn:EX_exponent_inequalities2} follows by setting $r(x)=s'\log P_{X}(x)$ in \eqref{eqn:Ex_type_MAP_equiv} which is an equivalent form of \eqref{eqn:Ex_type_MAP}.
\end{proof}

Corollary \ref{cor:EX_exponent_comparison} shows that both type-by-type and standard ensembles recover $E^{\rm tt}_{\mathrm{ex}}(R)$ with a family of metrics given by \eqref{eqn:Ex_Type_optimal_decoders} and \eqref{eqn:Ex_St_optimal_decoders}, where the earlier recovers the exponent for full range of rates with the same metric whilst latter recovers it with a different metric for each rate $R$.	
	
\subsection{Special case: no side information}
Specializing the result of Theorem \ref{th:standard} to block source coding without side information, we recover the following exponent for the case of using a mismatched metric $q(x)$ as
\begin{equation}
	E_{q}(R) = \sup_{\rho\geq 0,s\geq 0}\rho R-\log\sum\limits_{x\in\Xc}P_{X}(x)\bigg(\sum\limits_{\bar{x}\in\Xc}\bigg(\frac{q(\bar{x})}{q(x)}\bigg)^s\bigg)^{\rho}.
\end{equation}
For $q(x) = \frac{P_X(x)^{\tau}}{\sum\limits_{\bar x}P_X(\bar x)^{\tau}}$, $\tau>0$, which includes the matched decoding as special case, the above exponent recovers the exponent of the optimal code as \cite{jelinek1968probabilistic}
\begin{equation}
	E(R) = \sup_{\rho\geq 0}\rho R-\log\bigg(\sum\limits_{x\in\Xc}P_{X}(x)^{\frac{1}{1+\rho}}\bigg)^{1+\rho}.\label{eqn:no_side_info_matched_exponent}
\end{equation}

Theorem \ref{th:tbt} recovers \eqref{eqn:no_side_info_matched_exponent} independent of the employed decoding metric. In particular, by setting $a(x)=\frac{1}{1+\rho}\log P_X(x)$ and $s=0$, the type-by-type exponent recovers the exponent of the optimal source code \eqref{eqn:no_side_info_matched_exponent}.

\subsection{Achievable rates}

In this section, we derive the achievable rates for both standard and type-by-type random coding ensembles. Noticing that exponent is a convex function of the rate and the maximizing $\rho$ is the slope of the exponent curve, the achievable rates for both random coding schemes can be obtained similarly to \cite[Ch. 5]{gallager1968ita} by evaluating the partial derivative of the exponent to find the rate at which $\rho=0$ maximizes the exponent.

For the standard random coding case using \eqref{eqn:exp_St_RC} we find the achievable rate as
\begin{align}
	H_{q}(X|Y)=&\inf_{s\geq 0}-\sum\limits_{x\in\Xc,y\in\Yc}P_{XY}(x,y)\log\frac{q(x,y)^s}{\sum\limits_{\bar{x}\in\Xc}q(\bar{x},y)^s}\label{eq:rate_standard}\\
	=&H(X|Y)+\inf_{s\geq 0}D(P_{X|Y}\|Q^{(s)}_{X|Y}),	
\end{align}
where $Q^{(s)}_{X|Y}(x|y)=\frac{q(x,y)^s}{\sum\limits_{\bar{x}\in\Xc}q(\bar{x},y)^s}$.

Similarly for the type-by-type random coding case using \eqref{eqn:exp_Type_RC} we find the achievable rate as
\begin{align}
	H^{\rm tt}_{q}(X|Y)=&\inf_{s\geq 0,a(\cdot)}-\sum\limits_{x\in\Xc,y\in\Yc}P_{XY}(x,y)\log\frac{q(x,y)^se^{a(x)}}{\sum\limits_{\bar{x}\in\Xc}q(\bar{x},y)^se^{a(\bar{x})}}\label{eq:rate_tbt}\\
	=&H(X|Y)+\inf_{s\geq 0,a(\cdot)}D(P_{X|Y}\|Q^{(s,a(\cdot))}_{X|Y}),
\end{align}
where $Q^{(s,a(\cdot))}_{X|Y}(x|y)=\frac{q(x,y)^se^{a(x)}}{\sum\limits_{\bar{x}\in\Xc}q(\bar{x},y)^se^{a(\bar{x})}}$.

By inspecting the optimization problems in \eqref{eq:rate_standard} and \eqref{eq:rate_tbt} we have that 
\begin{equation}
H_{q}(X|Y) \geq H^{\rm tt}_{q}(X|Y).
\end{equation}
The expressions of the above achievable rates bear a strong resemblance to their channel coding counterparts, the generalized mutual information \cite{kaplan1993information} achieved by iid coding, and the LM rate \cite{huimis,Csi81} achieved by constant composition codes. Indeed, as shown next, these achievable rates can be expressed as a function of the generalized mutual information and the LM rate.

\begin{corollary}
	The achievable rate $H_{q}(X|Y)$ is related to the generalized mutual information \cite{kaplan1993information} of corresponding channel $W_{Y|X}(y|x)=\frac{P_{XY}(x,y)}{P_{X}(x)}$ with decoding metric $\bar{q}(x,y)=q(x,y)P(x)^{-\frac{1}{s^*}}$ as
	\begin{equation}
		H_{q}(X|Y)=H(X)-I^{\rm GMI}_{\bar{q}}(P_{X},W_{Y|X})\label{eqn:rate_standars_GMI_ralation}
	\end{equation}
	where $s^*$ is the optimizing parameter in \eqref{eq:rate_standard} and 
	\begin{equation}
I^{\rm GMI}_{\bar{q}}(P_{X},W_{Y|X}) = \sup_{s>0}\sum\limits_{x\in\Xc,y\in\Yc}P_{XY}(x,y)\log\frac{\bar{q}(x,y)^s}{\sum\limits_{\bar{x}\in\Xc}P_{X}(\bar{x})\bar{q}(\bar{x},y)^s}
\end{equation}
is the generalized mutual information \cite{kaplan1993information}.
\end{corollary}
\begin{proof}
	Denoting the objective in \eqref{eq:rate_standard} as $H_{q,s}(X|Y)$ we have
	\begin{align}
		H_{q,s}(X|Y)&=-\sum\limits_{x\in\Xc,y\in\Yc}P_{XY}(x,y)\log\frac{q(x,y)^s}{\sum\limits_{\bar{x}\in\Xc}q(\bar{x},y)^s}-\sum\limits_{x\in\Xc}P_{X}(x)\log P_{X}(x)+H(X)\\
		&=-\sum\limits_{x\in\Xc,y\in\Yc}P_{XY}(x,y)\log\frac{q(x,y)^s}{P_{X}(x)\sum\limits_{\bar{x}\in\Xc}q(\bar{x},y)^s}+H(X)\\
		&=-\sum\limits_{x\in\Xc,y\in\Yc}P_{XY}(x,y)\log\frac{\bar{q}(x,y)^s}{\sum\limits_{\bar{x}\in\Xc}P_{X}(\bar{x})\bar{q}(\bar{x},y)^s}+H(X)\label{eq:norm_metric_gmi}\\
		&=-I_{\bar{q},s}(P_{X},W_{Y|X})+H(X)
	\end{align}
	where in \eqref{eq:norm_metric_gmi} we have used the definition of decoding metric $\bar{q}(x,y)=q(x,y)P(x)^{-\frac{1}{s}}$.
	Taking the infimum of $H_{q,s}(X|Y)$ over $s$ and denoting the corresponding $s$ by $s^*$ we obtain \eqref{eqn:rate_standars_GMI_ralation}.
	
\end{proof}

\begin{corollary}
	The achievable rate $H^{\rm tt}_{q}(X|Y)$ is related to the LM rate \cite{huimis,Csi81} of corresponding channel as
	\begin{equation}
		H^{\rm tt}_{q}(X|Y)=H(X)-I^{\rm LM}_{q}(P_{X},W_{Y|X}).\label{eqn:rate_standars_LM_ralation}
	\end{equation}
	where 
	\begin{equation}
I^{\rm LM}_{q}(P_{X},W_{Y|X}) = \sup_{s>0, b(\cdot)}\sum\limits_{x\in\Xc,y\in\Yc}P_{XY}(x,y)\log\frac{q(x,y)^se^{b(x)}}{\sum\limits_{\bar{x}\in\Xc}P_{X}(\bar{x})q(\bar{x},y)^se^{b(\bar{x})}}
\end{equation}
is the LM rate \cite{huimis,Csi81}.

\end{corollary}
\begin{proof}
	Denoting the objective in \eqref{eq:rate_tbt} as $H^{\rm tt}_{q,s,a}(X|Y)$ we have
	\begin{align}
		H^{\rm tt}_{q,s,a}(X|Y)&=-\sum\limits_{x\in\Xc,y\in\Yc}P_{XY}(x,y)\log\frac{q(x,y)^se^{a(x)}}{\sum\limits_{\bar{x}\in\Xc}q(\bar{x},y)^se^{a(\bar{x})}}-\sum\limits_{x\in\Xc}P_{X}(x)\log P_{X}(x)+H(X)\\
		&=-\sum\limits_{x\in\Xc,y\in\Yc}P_{XY}(x,y)\log\frac{q(x,y)^se^{a(x)}}{P_{X}(x)\sum\limits_{\bar{x}\in\Xc}q(\bar{x},y)^se^{a(\bar{x})}}+H(X)\\
		&=-\sum\limits_{x\in\Xc,y\in\Yc}P_{XY}(x,y)\log\frac{q(x,y)^se^{b(x)}}{\sum\limits_{\bar{x}\in\Xc}P_{X}(\bar{x})q(\bar{x},y)^se^{b(\bar{x})}}+H(X)\label{eqn:replace_cost}\\
		&=-I_{q,s,b}(P_{X},W_{Y|X})+H(X),
	\end{align}
	where \eqref{eqn:replace_cost} follows from replacing $e^{a(x)}=P_{X}(x)e^{b(x)}$. Now taking the infimum of $H^{\rm tt}_{q,s,a}(X|Y)$ over $s$ and $a(\cdot)$ we obtain \eqref{eqn:rate_standars_LM_ralation}.
\end{proof}
The achievability of the rate in \eqref{eqn:rate_standars_LM_ralation}, expressed in the primal domain, was observed in \cite[Corollary 1]{Chen2009}.

\subsection{Numerical example}

We conclude this section with a numerical example. The joint distribution of the source $X$ with side information $Y$ is defined by the entries of the $|\Xc|\times|\Yc|$ matrix
\begin{equation}
	P_{XY}=\begin{bmatrix}
		0.588 & 0.006 & 0.006 \\
		0.03 & 0.24 & 0.03 \\
		0.02 & 0.02 & 0.06 
	\end{bmatrix}
	\label{eqn:example_joint_distribution}
\end{equation}
with $\Xc=\Yc=\{0,1,2\}$. We consider using a mismatched decoder with a memoryless metric given by the matrix
\begin{equation}
	q(x,y)=\begin{bmatrix}
		1-2\delta & \delta & \delta \\
		\delta & 1-2\delta & \delta \\
		\delta & \delta & 1-2\delta
	\end{bmatrix}
	\label{eqn:example_decoding_metic}
\end{equation}
with $\delta\in(0,\frac{1}{3})$. For any value of the $\delta\in(0,\frac{1}{3})$, the corresponding decoder is exactly the same and is equivalent to a minimum Hamming distance decoding metric. This is because for any $\delta$ and $\delta'$ in this range one can find an $s'\geq 0$ such that $(\frac{1-2\delta}{\delta})^{s}=(\frac{1-2\delta'}{\delta'})^{s'}$, so the optimization over $s$ makes the explicit choice of the $\delta$ irrelevant as long as it is in the range of $(0,\frac{1}{3})$.

Figure \ref{fig:exponent} illustrates the exponents for the standard and type-by-type ensembles with both matched and mismatched decoders. The sphere-packing upper bound is also shown for reference. We observe that the type-by-type expurgated exponent is higher for both matched and mismatched decoding. Indeed, as discussed earlier, in the matched case, the random coding components of $E_q(R)$ and $E^{\rm tt}_q(R)$ (for the rates below the critical rate $R_{\rm cr}$) both coincide with the sphere-packing upper-bound. The corresponding achievable rates are marked with dots in the figure. The conditional entropy $H(X|Y)=0.3879$ nats is the limit for the matched case while in the mismatched case we respectively have $H_{q}(X|Y)=0.4283$ nats and $H^{\rm tt}_{q}(X|Y)=0.4140$ nats. 
\begin{figure}
	\centering
	\input{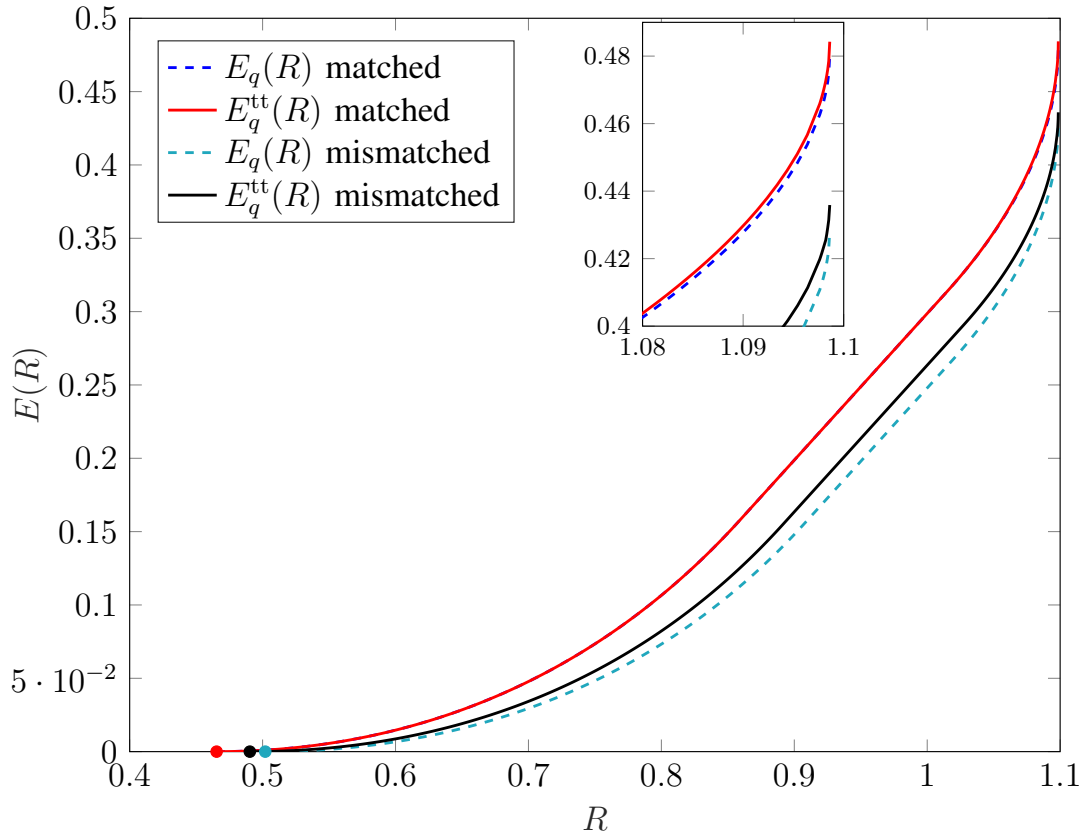}
	\caption{Error exponents for the source $X$ and side information $Y$ with joint distribution given in \eqref{eqn:example_joint_distribution}. The mismatched decoder uses the minimum Hamming distance metric.}
	\label{fig:exponent}
\end{figure}

\section{Expurgation method}
\label{sec:exp}

In this section we introduce the expurgation method for source coding that is valid for any source and side information model with arbitrary decoding metric. The method is based on randomly generating codes for the source and removing (expurgating) poor source sequences from the codes. We then show that, there exists a code in the expurgated ensemble, such that half of the source sequences meets a desired upper bound on the error probability. Repeating this procedure enough times we show that there exists a code in which all the source sequences meet the desired upper bound on the error probability. Details of the expurgation method are given in the following, first for the standard random coding and then for type-by-type random coding.

\subsection{Standard Block Coding}
\label{subsec:standard}

Given a standard block source code $\Cc$, the probability that decoder $\psi$ makes an error for a source sequence $\xb$ is given by
\begin{equation}
	p_e(\xb,\Cc)=\mathbb{P}\left[\bigcup_{\substack{\barxb\in\Xcn \\ \barxb\neq\xb}}\left\{q(\barxb,\Yb)\geq q(\xb,\Yb) , \phi(\barxb)=\phi(\xb)\right\}\right].
	\label{eqn:prob_union_error_events}
\end{equation}

Considering the standard random coding ensemble $\Csf$ defined in Section \ref{sec:intro}, $p_e(\xb,\Csf)$ is a random variable for every $\xb$. 
Applying Markov's inequality to this random variable, we obtain
\begin{equation}
	\mathbb{P}\Big[p_e(\xb,\Csf)^{\frac{1}{\rho}}\geq \eta\mathbb{E}\Big[p_e(\xb,\Csf)^{\frac{1}{\rho}}\Big]\Big]
	\leq\frac{1}{\eta},
\end{equation}
for any $\rho>0$ and $\eta>1$. Equivalently we have
\begin{equation}
	\mathbb{P}\Big[p_e(\xb,\Csf)^{\frac{1}{\rho}}\leq \eta\mathbb{E}\Big[p_e(\xb,\Csf)^{\frac{1}{\rho}}\Big]\Big]
	\geq1-\frac{1}{\eta}.\label{eqn:bound_Markov_inequality}
\end{equation}

For a given standard block source code $\Cc$, denote by $N_0(\Cc)$ the number of source sequences $\xb\in\Xcn$ that satisfy 
\begin{equation}
	p_e(\xb,\Cc)^{\frac{1}{\rho}}\leq \eta\mathbb{E}\Big[p_e(\xb,\Csf)^{\frac{1}{\rho}}\Big].\label{eqn:bound_sequence_error}
\end{equation}
Over the standard random coding ensemble $\Csf$, $N_0(\Csf)$ is a random variable.
Similarly to \cite[p. 151]{gallager1968ita}, we first observe that the expected value of $N_0(\Csf)$ can be lower bounded as
\begin{align}
	\mathbb{E}\big[N_0(\Csf)\big]&= \EE\bigg[\sum\limits_{\xb\in\Xcn}\indicator\Big\{p_e(\xb,\Csf)^{\frac{1}{\rho}}\leq \eta\mathbb{E}\Big[p_e(\xb,\Csf)^{\frac{1}{\rho}}\Big]\Big\}\bigg]\\
	&=\sum\limits_{\xb\in\Xcn}\EE\Big[\indicator\Big\{p_e(\xb,\Csf)^{\frac{1}{\rho}}\leq \eta\mathbb{E}\Big[p_e(\xb,\Csf)^{\frac{1}{\rho}}\Big]\Big\}\Big]\\
	&=\sum\limits_{\xb\in\Xcn}\PP\Big[p_e(\xb,\Csf)^{\frac{1}{\rho}}\leq \eta\mathbb{E}\Big[p_e(\xb,\Csf)^{\frac{1}{\rho}}\Big]\Big]\\
	&\geq |\Xc|^n\bigg(1-\frac{1}{\eta}\bigg)\label{eqn:bound_num_seq}
\end{align}
where the last step follows from \eqref{eqn:bound_Markov_inequality}.

Now we choose $\eta=2$ which yields $1-\frac{1}{\eta}=\frac{1}{2}$ in \eqref{eqn:bound_num_seq}. Therefore, we conclude that there exists a standard block code $\Cc$ in which at least half of the source sequences satisfy the inequality
\begin{equation}
	p_e(\xb,\Cc)\leq \left(2\mathbb{E}\left[p_e(\xb,\Csf)^{\frac{1}{\rho}}\right]\right)^{\rho}.\label{eqn:bound_seq_error}
\end{equation}
Let us denote the set of above mentioned sequences as $\Ac_1$ where $\Ac_1\subset\Xcn$ and $|\Ac_1|\geq\frac{|\Xcn|}{2}$.

We now construct a code $\Cc_1$ for the set $\Ac_1$ from the code $\Cc$ as follows:
We keep all the sequences $\xb\in\Ac_1$ in $\Cc$, which satisfy \eqref{eqn:bound_seq_error}, and we expurgate rest of the sequences $\xb\in\Xcn\backslash\Ac_1$ from code $\Cc$. Note that removing source sequences can only reduce the error probability of remaining ones.
This construction yields a code $\Cc_1$ which includes only the sequences in $\Ac_1$ and for every $\xb\in\Ac_1$ we have 
\begin{equation}
	p_e(\xb,\Cc_1)\leq \left(2\mathbb{E}\left[p_e(\xb,\Csf)^{\frac{1}{\rho}}\right]\right)^{\rho}.\label{eqn:bound_original_error}
\end{equation}

The rest of sequences in $\Xcn\backslash\Ac_1$ include at most half of the original sequences. We now repeat the procedure to find a good code from the standard random coding ensemble for the sequences in $\Xcn\backslash\Ac_1$ (which we denote for simplicity as $\bar{\Csf}$). Given a code $\bar{\Cc}$ from the ensemble $\bar{\Csf}$, the probability of decoding error for a sequence $\xb\in\Xcn\backslash\Ac_1$ is given by \eqref{eqn:prob_union_error_events} with the only difference that now the corresponding union is over $\barxb\in\Xcn\backslash\Ac_1$ instead of $\xb\in\Xcn$. Therefore, the average error probability of a source sequence in the new ensemble is upper bounded by that of the original ensemble, i.e., for any $\xb\in\Xcn\backslash\Ac_1$ we have
\begin{equation}
	\mathbb{E}\Big[p_e(\xb,\bar{\Csf})^{\frac{1}{\rho}}\Big]\leq\mathbb{E}\Big[p_e(\xb,\Csf)^{\frac{1}{\rho}}\Big].
\end{equation}
Now following similar steps as before we show existence of a code $\bar\Cc$ from the ensemble $\bar{\Csf}$ in which at least half of the remaining $\Xcn\backslash\Ac_1$ source sequences satisfy the inequality
\begin{align}
	p_e(\xb,\bar{\Cc})&\leq \left(2\mathbb{E}\left[p_e(\xb,\bar{\Csf})^{\frac{1}{\rho}}\right]\right)^{\rho}\\
	&\leq\left(2\mathbb{E}\left[p_e(\xb,\Csf)^{\frac{1}{\rho}}\right]\right)^{\rho}.\label{eqn:bound_seq_error_2iter}
\end{align}

We denote by $\Ac_2$ the set of above mentioned sequences satisfying \eqref{eqn:bound_seq_error_2iter} where $\Ac_2\subset\Xcn\backslash\Ac_1$ and $|\Ac_2|\geq\frac{|\Xcn|-|\Ac_1|}{2}$. We now construct a code $\Cc_2$ for the set $\Ac_2$ from the code $\bar\Cc$ following similar expurgation step described above. This construction yields a code $\Cc_2$ which includes only the sequences in $\Ac_2$ and every sequence $\xb\in\Ac_2$ satisfies the desired bound as in \eqref{eqn:bound_original_error}.

The set $\Xcn\backslash\Ac_1\bigcup\Ac_2$ includes at most one forth of the original sequences. By repeating this procedure $k$ times, the set $\Xcn\backslash\bigcup_{i=1}^k\Ac_i$ includes at most a fraction $\frac{1}{2^k}$ of the original sequences. By choosing $k=n\log_{2}|\Xc|$ we guarantee that $\bigcup_{i=1}^k\Ac_i=\Xcn$. Finally we construct the code $\Cc_{\rm ex}$ by combining all $\Cc_i$'s as $\Cc_{\rm ex}=\{\Cc_1,\cdots,\Cc_k\}$. The code $\Cc_{\rm ex}$ includes all the source sequences and every sequence satisfies the desired bound as in \eqref{eqn:bound_original_error}. Considering the number of codewords in each iteration as $\frac{M}{k}$, the total number of codewords of the final code $\Cc_{\rm ex}$ is $M$.

Therefore we proved the following lemma.

\begin{lemma}
	There exists a code $\Cc_{\rm ex}$ in the standard random coding ensemble such that for every source sequence $\xb\in\Xcn$ and $\rho\geq 0$
	\begin{equation}
		p_e(\xb, \Cc_{ex \rm})\leq\Big(2\mathbb{E}\left[p_e(\xb,\Csf)^{\frac{1}{\rho}}\right]\Big)^{\rho}\label{eqn:bound_actual_seq_St},
	\end{equation}
	where the expectation is over the standard random coding ensemble with $\frac{M}{k}$ codewords with $k=n\log_{2}|\Xc|$.
	\label{Lem:expurgation_St}
\end{lemma}

\subsection{Type-by-Type Coding}
\label{subsec:type-by-type}

Given a type-by-type block source code $\Ccb$, the probability that decoder $\psi_i$ yields an error for a source sequence $\xb\in\Tc_n(\hat{P}_i)$ is given by
\begin{equation}
	p_e(\xb,\Cc_i)=\mathbb{P}\left[\bigcup_{\substack{\barxb\in\Tc_n(\hat{P}_i)\\\barxb\neq\xb}}\left\{q(\barxb,\Yb)\geq q(\xb,\Yb) , \phi(\barxb)=\phi(\xb)\right\}\right].
\end{equation}

We notice that the type-by-type ensemble $\Csfb$ is a product of random coding ensembles for each type $\Csf_i$ and we can use the expurgation argument independently for each type. Here we only focus on a given fixed type $\hat{P}_i$. Applying Markov's inequality to random variable $p_e(\xb,\Csf_i)$, we obtain for every $\xb\in\Tc_n(\hat{P}_i)$
\begin{equation}
	\mathbb{P}\Big[p_e(\xb,\Csf_i)^{\frac{1}{\rho}}\geq \eta\mathbb{E}\Big[p_e(\xb,\Csf_i)^{\frac{1}{\rho}}\Big]\Big]
	\leq\frac{1}{\eta},
\end{equation}
for any $\rho>0$ and $\eta>1$. Equivalently we have
\begin{equation}
	\mathbb{P}\Big[p_e(\xb,\Csf_i)^{\frac{1}{\rho}}\leq \eta\mathbb{E}\Big[p_e(\xb,\Csf_i)^{\frac{1}{\rho}}\Big]\Big]
	\geq1-\frac{1}{\eta}.\label{eqn:bound_Markov_inequality_type}
\end{equation}

For a given type-by-type code $\Cc_i$, denote by $N_0(\hat{P}_i,\Cc_i)$ the number of source sequences $\xb\in\Tc_n(\hat{P}_i)$ that satisfy 
\begin{equation}
	p_e(\xb,\Cc_i)^{\frac{1}{\rho}}\leq \eta\mathbb{E}\Big[p_e(\xb,\Csf_i)^{\frac{1}{\rho}}\Big].\label{eqn:bound_sequence_error_type}
\end{equation}
Over the random ensemble $\Csf_i$, $N_0(\hat{P}_i,\Csf_i)$ is a random variable.
Now choosing $\eta=2$, and using \eqref{eqn:bound_Markov_inequality_type} the expected value of $N_0(\hat{P}_i,\Csf_i)$ is bounded below as
\begin{equation}
	\mathbb{E}\big[N_0(\hat{P}_i,\Csf_i)\big]\geq \frac{|\Tc_n(\hat{P}_i)|}{2}.
\end{equation}
Therefore, we conclude that there exists a source code $\Cc_i$ for type $\hat{P}_i$, such that at least half of the sequences of that type satisfy the inequality
\begin{equation}
	p_e(\xb,\Cc_i)\leq \left(2\mathbb{E}\left[p_e(\xb,\Csf_i)^{\frac{1}{\rho}}\right]\right)^{\rho}.\label{eqn:bound_seq_error_type}
\end{equation}
Let us denote the set of above mentioned sequences as $\Ac_{i1}$ where $\Ac_{i1}\subset|\Tc_n(\hat{P}_i)|$ and $|\Ac_{i1}|\geq\frac{|\Tc_n(\hat{P}_i)|}{2}$.

We now construct a code $\Cc_{i1}$ for the set $\Ac_{i1}$ from the code $\Cc_i$ as follows:
We keep all the sequences $\xb\in\Ac_{i1}$ in $\Cc_i$, which satisfy \eqref{eqn:bound_seq_error_type}, and we expurgate rest of the sequences $\xb\in\Tc_n(\hat{P}_i)\backslash\Ac_{i1}$ from code $\Cc_i$. Note that removing source sequences can only reduce the error probability of remaining ones.

The rest of sequences in $\Tc_n(\hat{P}_i)\backslash\Ac_{i1}$ include at most half of the sequences from type $\hat{P}_i$. We now repeat the procedure by considering the random coding ensemble for the remaining sequences $\Tc_n(\hat{P}_i)\backslash\Ac_{i1}$ (denoted by $\bar{\Csf}_i$) and finding a good code $\bar{\Cc}_i$ from this ensemble and constructing a code $\Cc_{i2}$ for the set $\Ac_{i2}$. Similarly to the analysis in \ref{subsec:standard} we apply this procedure successively for at most $k_i=\log_{2}|\Tc_n(\hat{P}_i)|\leq n\log_{2}|\Xc|$ times and we guarantee that $\bigcup_{j=1}^k\Ac_{ij}=\Tc_n(\hat{P}_i)$. Finally we construct the code $\Cc_{{\rm ex},i}$ by combining all codes $\Cc_{ij}$ as $\Cc_{{\rm ex},i}=\{\Cc_{i1},\cdots,\Cc_{ik_i}\}$. The code $\Cc_{{\rm ex},i}$ includes all the source sequences in $\Tc_n(\hat{P}_i)$ and every sequence satisfies the desired bound as in \eqref{eqn:bound_seq_error_type}. Considering the number of codewords in each iteration as $\frac{M}{k_i|\Pc_n(\Xc)|}$, the total number of codewords of the code $\Cc_{{\rm ex},i}$ is $\frac{M}{|\Pc_n(\Xc)|}$. In a similar way we obtain the code for every type and combine those to construct the type-by-type code $\Ccb_{\rm ex}$. 

\begin{lemma}
	There exists a code $\Ccb_{\rm ex}=\{\Cc_{{\rm ex},1},\cdots,\Cc_{{\rm ex},|\Pc_n(\Xc)|}\}$ in the type-by-type random coding ensemble such that for every $i\in\{1,\cdots,|\Pc_n(\Xc)|\}$ and for every source sequence $\xb\in\Tc_n(\hat{P}_i)$ and $\rho\geq 0$
	\begin{equation}
		p_e(\xb,\Ccb_{\rm ex})\leq\Big(2\mathbb{E}\left[p_e(\xb,\Csf_i)^{\frac{1}{\rho}}\right]\Big)^{\rho}\label{eqn:bound_actual_seq_Type},
	\end{equation}
	where the expectation is over the random coding ensemble for the corresponding type with $\frac{M}{k_i|\Pc_n(\Xc)|}$ codewords and $k_i=\log_{2}|\Tc_n(\hat{P}_i)|\leq n\log_{2}|\Xc|$.
	\label{Lem:expurgation_Type}
\end{lemma}

\section{Conclusions}
We have introduced an expurgation method for source coding that is valid for general source and side information models and arbitrary decoding metrics. Building on the developed expurgation method, we have derived dual-domain multi-letter upper bounds on the error probability of an expurgated code. We have further specialized the bounds to memoryless source models and memoryless mismatched decoding metrics and derived two achievable exponents for the standard and type-by-type random coding ensembles; the latter is shown to coincide with the Csisz\'ar and K\"orner's exponent. We have shown existence of a code in each ensemble that achieves the maximum of the corresponding random coding and expurgated exponent. While the expurgated exponent of the standard ensemble is shown to be weaker of the two, its derivation is relatively simpler and does not rely on the method of types.

\section{Proof of Theorem \ref{Thm:Exponent_St}}
\label{sec:proof_th1}

Lemma \ref{Lem:expurgation_St} showed the existence of a good code in the standard random coding ensemble which satisfies an upper bound on the error probability for every source sequence. Here we show that this code achieves the exponent of Theorem \ref{Thm:Exponent_St}. We first show the achievability of $E_{q,\mathrm{ex}}(R)$ and then $E_{q,\mathrm{r}}(R)$.

We start by bounding $p_e(\xb,\Cc)$ for a given source sequence $\xb$ and given standard code $\Cc$ as follows

\begin{align}
	p_e(\xb,\Cc)&=\mathbb{P}\left[\bigcup_{\substack{\barxb\in\Xcn \\ \barxb\neq\xb}}\left\{q(\barxb,\Yb)\geq q(\xb,\Yb) , \phi(\barxb)=\phi(\xb)\right\}\right]\\
	&=\sum\limits_{\yb\in\Ycn}P_{\Yb|\Xb}(\yb|\xb)\mathds{1}\left[\bigcup_{\substack{\barxb\in\Xcn \\ \barxb\neq\xb}}\left\{q(\barxb,\yb)\geq q(\xb,\yb) , \phi(\barxb)=\phi(\xb)\right\}\right]\\
	&\leq\sum\limits_{\yb\in\Ycn}P_{\Yb|\Xb}(\yb|\xb)\sum\limits_{\substack{\barxb\in\Xcn \\ \barxb\neq\xb}}\mathds{1}\left[q(\barxb,\yb)\geq q(\xb,\yb) , \phi(\barxb)=\phi(\xb)\right]\label{eqn:union_bound}\\
	&=\sum\limits_{\yb\in\Ycn}P_{\Yb|\Xb}(\yb|\xb)\sum\limits_{\substack{\barxb\in\Xcn \\ \barxb\neq\xb}}\mathds{1}\left[q(\barxb,\yb)\geq q(\xb,\yb) \right] \mathds{1}\left[\phi(\barxb)=\phi(\xb)\right]\\
	&\leq\sum\limits_{\substack{\barxb\in\Xcn \\ \barxb\neq\xb}}\mathds{1}\left[\phi(\barxb)=\phi(\xb)\right]\sum\limits_{\yb\in\Ycn}P_{\Yb|\Xb}(\yb|\xb)\left(\frac{q(\barxb,\yb)}{q(\xb,\yb)}\right)^{s}\label{eqn:introduce_s},
\end{align}
where we use union bound in \eqref{eqn:union_bound} and \eqref{eqn:introduce_s} holds for any $s\geq0$.

Now considering the ensemble of random standard block source codes and denoting the induced random encoding function by $\Phi(\cdot)$, we upper bound the $\mathbb{E}\left[p_e(\xb,\Csf)^{\frac{1}{\rho}}\right]$ using \eqref{eqn:introduce_s} as follows

\begin{align}
	\mathbb{E}\left[p_e(\xb,\Csf)^{\frac{1}{\rho}}\right]&\leq\mathbb{E}\left[\left(\sum\limits_{\substack{\barxb\in\Xcn \\ \barxb\neq\xb}}\mathds{1}\left[\Phi(\barxb)=\Phi(\xb)\right]\sum\limits_{\yb\in\Ycn}P_{\Yb|\Xb}(\yb|\xb)\left(\frac{q(\barxb,\yb)}{q(\xb,\yb)}\right)^{s}\right)^{\frac{1}{\rho}}\right]\\
	&\leq\sum\limits_{\barxb\in\Xcn}\mathbb{E}\left[\mathds{1}\left[\Phi(\barxb)=\Phi(\xb)\right]\right]\left(\sum\limits_{\yb\in\Ycn}P_{\Yb|\Xb}(\yb|\xb)\left(\frac{q(\barxb,\yb)}{q(\xb,\yb)}\right)^s\right)^{\frac{1}{\rho}}\label{eqn:introduce_inequality_rho}\\
	&\leq\frac{k}{M}\sum\limits_{\barxb\in\Xcn}\left(\sum\limits_{\yb\in\Ycn}P_{\Yb|\Xb}(\yb|\xb)\left(\frac{q(\barxb,\yb)}{q(\xb,\yb)}\right)^s\right)^{\frac{1}{\rho}}\label{eqn:convenient_upperbound}
\end{align}
where \eqref{eqn:introduce_inequality_rho} follows from inequality $(\sum\limits_i a_i)^{\frac{1}{\rho}}\leq\sum\limits_i a_i^{\frac{1}{\rho}}$ for $\rho\geq 1$ and including $\barxb=\xb$ in the summation, \eqref{eqn:convenient_upperbound} follows from $\mathbb{E}\left[\mathds{1}\left[\Phi(\barxb)=\Phi(\xb)\right]\right]\leq\frac{k}{M}$ where $k=n\log_2|\Xc|$ is an upper bound on the number of iterations in expurgation method.

Now substituting \eqref{eqn:convenient_upperbound} in \eqref{eqn:bound_actual_seq_St} from Lemma \ref{Lem:expurgation_St} and summing over all source sequences we find an upper bound on the error probability of the codebook $\Cc_{\rm ex}$ as 

\begin{align}
	p_e(\Cc_{\rm ex})&=\sum\limits_{\xb\in\Xcn}P_{\Xb}(\xb)p_e(\xb,\Cc_{\rm ex})\\
	&\leq\sum\limits_{\xb\in\Xcn}P_{\Xb}(\xb)\Big(2\mathbb{E}\left[p_e(\xb,\Csf)^{\frac{1}{\rho}}\right]\Big)^{\rho}\\
	&\leq\left(\frac{2k}{M}\right)^{\rho}\sum\limits_{\xb\in\Xcn}\left(\sum\limits_{\barxb\in\Xcn}\left(\sum\limits_{\yb\in\Ycn}P_{\Xb\Yb}(\xb,\yb)\left(\frac{q(\barxb,\yb)}{q(\xb,\yb)}\right)^s\right)^{\frac{1}{\rho}}\right)^{\rho}\label{eqn:bound_expurgated}
\end{align}

Equation \eqref{eqn:bound_expurgated} is valid for any discrete source and any decoding metric. We now specialize this to the case of memoryless sources and metrics, using $q(\xb,\yb)=\prod_{i=1}^{n}q(x_i,y_i)$. Therefore we obtain
\begin{align}
	p_e(\Cc_{\rm ex})&\leq\left(\frac{2k}{M}\right)^{\rho}\left(\sum\limits_{x\in\Xc}\left(\sum\limits_{\bar{x}\in\Xc}\left(\sum\limits_{y\in\Yc}P_{XY}(x,y)\left(\frac{q(\bar{x},y)}{q(x,y)}\right)^s\right)^{\frac{1}{\rho}}\right)^{\rho}\right)^{n}\\
	&=e^{-n(\rho (R-\delta_n)-E_{\rm{x}}(\rho,s))}
\end{align}
where $R=\frac{\log M}{n}$, $\delta_n=\frac{\log(2n\log_2|\Xc|)}{n}\to 0$ as $n\to\infty$ and
\begin{equation}
	E_{\rm{x}}(\rho,s)=\log\sum\limits_{x\in\Xc}\left(\sum\limits_{\bar{x}\in\Xc}\left(\sum\limits_{y\in\Yc}P_{XY}(x,y)\left(\frac{q(\bar{x},y)}{q(x,y)}\right)^s\right)^{\frac{1}{\rho}}\right)^{\rho},\label{eqn:Ex_St}
\end{equation}

Hence, we show the achievability of the exponent $E_{q,\mathrm{ex}}(R)$ by optimizing over $\rho,s$ as
\begin{equation}
	E_{q,\mathrm{ex}}(R)=\max_{\rho\geq 1,s\geq 0}\rho R-E_{\rm{x}}(\rho,s).\label{eqn:expurgated_exponent}
\end{equation}

To show that the same code of Lemma \ref{Thm:Exponent_St} achieves $E_{q,\mathrm{r}}(R)$ we use \eqref{eqn:bound_actual_seq_St} of Lemma \ref{Lem:expurgation_St} with $\rho=1$, and obtain

\begin{align}
	p_e(\xb,\Cc_{\rm ex})\leq2\mathbb{E}\left[P_e(\xb,\Csf)\right].\label{eqn:special_case_st}
\end{align}

Averaging over all source sequences we obtain
\begin{align}
	p_e(\Cc_{\rm ex})&=\sum\limits_{\xb\in\Xcn}P_{\Xb}(\xb)p_e(\xb,\Cc_{\rm ex})\\
	&\leq2\sum\limits_{\xb\in\Xcn}P_{\Xb}(\xb)\mathbb{E}\left[P_e(\xb,\Csf)\right]\\
	&=2\mathbb{E}\left[\sum\limits_{\xb\in\Xcn}P_{\Xb}(\xb)P_e(\xb,\Csf)\right]\\
	&=2\mathbb{E}\left[p_e(\Csf)\right],\label{eqn:ensemble_average}
\end{align}
which shows that the error probability of $\Cc_{\rm ex}$ is upper bounded by twice the ensemble average error probability.

Now we derive an upper bound on the ensemble average error probability following the derivations in \cite{gallager1979source}.

We start by bounding $p_e(\yb,\Cc)$ for a given side information sequence $\yb$ and given standard code $\Cc$ as follows

\begin{align}
	p_e(\yb,\Cc)&=\sum\limits_{\xb\in\Xcn}P_{\Xb|\Yb}(\xb|\yb)\mathds{1}\left[\bigcup_{\substack{\barxb\in\Xcn \\ \barxb\neq\xb}}\left\{q(\barxb,\yb)\geq q(\xb,\yb) , \phi(\barxb)=\phi(\xb)\right\}\right]\\
	&\leq\sum\limits_{\xb\in\Xcn}P_{\Xb|\Yb}(\xb|\yb)\left(\sum_{\substack{\barxb\in\Xcn \\ \barxb\neq\xb}}\mathds{1}\left[q(\barxb,\yb)\geq q(\xb,\yb) , \phi(\barxb)=\phi(\xb)\right]\right)^\rho\label{eqn:use_inequality}\\
	&=\sum\limits_{\xb\in\Xcn}P_{\Xb|\Yb}(\xb|\yb)\left(\sum_{\substack{\barxb\in\Xcn \\ \barxb\neq\xb}}\mathds{1}\left[q(\barxb,\yb)\geq q(\xb,\yb)\right]\mathds{1}\left[ \phi(\barxb)=\phi(\xb)\right]\right)^\rho\\
	&\leq\sum\limits_{\xb\in\Xcn}P_{\Xb|\Yb}(\xb|\yb)\left(\sum_{\substack{\barxb\in\Xcn \\ \barxb\neq\xb}}\left(\frac{q(\barxb,\yb)}{q(\xb,\yb)}\right)^{s}\mathds{1}\left[ \phi(\barxb)=\phi(\xb)\right]\right)^\rho,\label{eqn:bound_with_s}
\end{align}
where \eqref{eqn:use_inequality} follows from using the inequality $\mathds{1}\left[\bigcup\limits_{i}A_i\right]\leq\left(\sum\limits_{i}\mathds{1}\left[A_{i}\right]\right)^{\rho}$ for any set of events $\{A_i\}$ and $\rho\in[0,1]$ and \eqref{eqn:bound_with_s} holds for any $s\geq0$.

Now considering the ensemble of random standard block source codes we upper bound the $\mathbb{E}\left[p_e(\yb,\Csf)\right]$ using \eqref{eqn:bound_with_s} as follows
\begin{align}
	\mathbb{E}\left[p_e(\yb,\Csf)\right]&\leq\sum\limits_{\xb\in\Xcn}P_{\Xb|\Yb}(\xb|\yb)\mathbb{E}\left[\left(\sum_{\substack{\barxb\in\Xcn \\ \barxb\neq\xb}}\left(\frac{q(\barxb,\yb)}{q(\xb,\yb)}\right)^{s}\mathds{1}\left[ \Phi(\barxb)=\Phi(\xb)\right]\right)^\rho\right]\\
	&\leq\sum\limits_{\xb\in\Xcn}P_{\Xb|\Yb}(\xb|\yb)\left(\sum_{\substack{\barxb\in\Xcn \\ \barxb\neq\xb}}\left(\frac{q(\barxb,\yb)}{q(\xb,\yb)}\right)^{s}\mathbb{E}\left[\mathds{1}\left[ \Phi(\barxb)=\Phi(\xb)\right]\right] \right)^\rho\label{eqn:jensen_ineq}\\
	&\leq\left(\frac{k}{M}\right)^{\rho}\sum\limits_{\xb\in\Xcn}P_{\Xb|\Yb}(\xb|\yb)\left(\sum_{\barxb\in\Xcn}\left(\frac{q(\barxb,\yb)}{q(\xb,\yb)}\right)^{s}\right)^\rho\label{eqn:bound_fixed_side_information},
\end{align}
where \eqref{eqn:jensen_ineq} follows from Jensen's inequality and the concavity of $x^{\rho}$ for $\rho\in\left[0,1\right]$ and \eqref{eqn:bound_fixed_side_information} follows from $\mathbb{E}\left[\mathds{1}\left[ \Phi(\barxb)=\Phi(\xb)\right]\right]\leq\frac{k}{M}$ and including $\barxb=\xb$ in the summation.

Averaging over all side information sequences we obtain an upper bound on the ensemble average error probability as
\begin{align}
	\mathbb{E}\left[p_e(\Csf)\right]&=\sum\limits_{\yb\in\Ycn}P_{\Yb}(\yb)\mathbb{E}\left[p_e(\yb,\Csf)\right]\\
	&\leq\left(\frac{k}{M}\right)^{\rho}\sum\limits_{\xb\in\Xcn,\yb\in\Ycn}P_{\Xb\Yb}(\xb,\yb)\left(\sum_{\barxb\in\Xcn}\left(\frac{q(\barxb,\yb)}{q(\xb,\yb)}\right)^{s}\right)^\rho\label{eqn:random_coding_bound}
\end{align}
where \eqref{eqn:random_coding_bound} holds for any $\rho\in\left[0,1\right]$ and $s\geq0$. Introducing \eqref{eqn:random_coding_bound} in \eqref{eqn:ensemble_average} and particularizing it to the case of memoryless sources and metrics we obtain
\begin{align}
	p_e(\Cc_{\rm ex})&\leq2\left(\frac{k}{M}\right)^{\rho}\left(\sum\limits_{x\in\Xc,y\in\Yc}P_{XY}(x,y)\left(\sum\limits_{\bar{x}\in\Xc}\left(\frac{q(\bar{x},y)}{q(x,y)}\right)^s\right)^{\rho}\right)^{n}\\
	&=e^{-n(\rho (R-\delta_n)-E_{\rm{s}}(\rho,s)-\delta'_n)}
\end{align}
where $\delta_n=\frac{\log k}{n}\to 0$ and $\delta'_n=\frac{\log 2}{n}\to 0$ as $n\to\infty$ and 
\begin{equation}
	E_{\rm{s}}(\rho,s)=\log\sum\limits_{x\in\Xc,y\in\Yc}P_{XY}(x,y)\left(\sum\limits_{\bar{x}\in\Xc}\left(\frac{q(\bar{x},y)}{q(x,y)}\right)^s\right)^{\rho}
\end{equation}
Hence, we show the achievability of the exponent $E_{q,\mathrm{r}}(R)$ by optimizing over $\rho,s$ as in \eqref{eqn:exp_St_RC}.

\section{Proof of Theorem \ref{Thm:Exponent_Type}}
\label{sec:proof_th2}

Lemma \ref{Lem:expurgation_Type} showed the existence of a good code in the type-by-type random coding ensemble which satisfies an upper bound on the error probability for every source sequence. Here we show that this code achieves the exponent of Theorem \ref{Thm:Exponent_Type}. As for the proof in the previous section, we first show the achievability of $E^{\rm tt}_{q,\mathrm{ex}}(R)$ and then that of $E^{\rm tt}_{q,\mathrm{r}}(R)$.

We start by bounding $P_e(\xb,\Cc_i)$ for a given source sequence $\xb\in\Tc_n(\hat{P}_{i})$ and given type-by-type code $\Ccb$ with subcode $\Cc_i$ for type $\hat{P}_{i}$ as follows
\begin{align}
	P_e(\xb,\Cc_i)&=\mathbb{P}\left[\bigcup_{\substack{\barxb\in\Tc_n(\hat{P}_i)\\\barxb\neq\xb}}\left\{q(\barxb,\Yb)\geq q(\xb,\Yb) , \phi_i(\barxb)=\phi_i(\xb)\right\}\right]\\
	&=\sum\limits_{\yb\in\Ycn}P_{\Yb|\Xb}(\yb|\xb)\mathds{1}\left[\bigcup_{\substack{\barxb\in\Tc_n(\hat{P}_i)\\\barxb\neq\xb}}\left\{q(\barxb,\Yb)\geq q(\xb,\Yb) , \phi_i(\barxb)=\phi_i(\xb)\right\}\right]\\
	&\leq\sum\limits_{\yb\in\Ycn}P_{\Yb|\Xb}(\yb|\xb)\sum\limits_{\substack{\barxb\in\Tc_n(\hat{P}_i)\\\barxb\neq\xb}}\mathds{1}\left[q(\barxb,\yb)\geq q(\xb,\yb) , \phi_i(\barxb)=\phi_i(\xb)\right]\label{eqn:union_bound_type}\\
	&=\sum\limits_{\yb\in\Ycn}P_{\Yb|\Xb}(\yb|\xb)\sum\limits_{\substack{\barxb\in\Tc_n(\hat{P}_i)\\\barxb\neq\xb}}\mathds{1}\left[q(\barxb,\yb)\geq q(\xb,\yb) \right] \mathds{1}\left[\phi_i(\barxb)=\phi_i(\xb)\right]\\
	&\leq\sum\limits_{\substack{\substack{\barxb\in\Tc_n(\hat{P}_i)\\\barxb\neq\xb}}}\mathds{1}\left[\phi_i(\barxb)=\phi_i(\xb)\right]\sum\limits_{\yb\in\Ycn}P_{\Yb|\Xb}(\yb|\xb)\left(\frac{q(\barxb,\yb)}{q(\xb,\yb)}\right)^{s}\label{eqn:introduce_s_type}
\end{align}
where we use union bound in \eqref{eqn:union_bound_type} and \eqref{eqn:introduce_s_type} holds for any $s\geq0$.

Now considering the ensemble of random type-by-type codes and denoting the induced random encoding function by $\Phi_i(\cdot)$, we upper bound the $\mathbb{E}\left[P_e(\xb,\Csf_i)^{\frac{1}{\rho}}\right]$ using \eqref{eqn:introduce_s_type} as follows 

\begin{align}
	\mathbb{E}\left[P_e(\xb,\Csf_i)^{\frac{1}{\rho}}\right]&\leq\mathbb{E}\left[\left(\sum\limits_{\substack{\substack{\barxb\in\Tc_n(\hat{P}_i)\\\barxb\neq\xb}}}\mathds{1}\left[\Phi_i(\barxb)=\Phi_i(\xb)\right]\sum\limits_{\yb\in\Ycn}P_{\Yb|\Xb}(\yb|\xb)\left(\frac{q(\barxb,\yb)}{q(\xb,\yb)}\right)^{s}\right)^{\frac{1}{\rho}}\right]\\
	&\leq\sum\limits_{\barxb\in\Tc_n(\hat{P}_i)}\mathbb{E}\left[\mathds{1}\left[\Phi_i(\barxb)=\Phi_i(\xb)\right]\right]\left(\sum\limits_{\yb\in\Ycn}P_{\Yb|\Xb}(\yb|\xb)\left(\frac{q(\barxb,\yb)}{q(\xb,\yb)}\right)^s\right)^{\frac{1}{\rho}}\label{eqn:introduce_inequality_rho_type}\\
	&\leq\frac{k_i|\Pc_n(\Xc)|}{M}\sum\limits_{\barxb\in\Tc_n(\hat{P}_i)}\left(\sum\limits_{\yb\in\Ycn}P_{\Yb|\Xb}(\yb|\xb)\left(\frac{q(\barxb,\yb)}{q(\xb,\yb)}\right)^s\right)^{\frac{1}{\rho}}\label{eqn:convenient_upperbound_type}
\end{align}
where \eqref{eqn:introduce_inequality_rho_type} follows from inequality $(\sum\limits_i a_i)^{\frac{1}{\rho}}\leq\sum\limits_i a_i^{\frac{1}{\rho}}$ for $\rho\geq 1$ and including $\barxb=\xb$ in the summation, \eqref{eqn:convenient_upperbound_type} follows from $\mathbb{E}\left[\mathds{1}\left[\Phi_i(\barxb)=\Phi_i(\xb)\right]\right]\leq\frac{k_i|\Pc_n(\Xc)|}{M}$ where $k_i=\log_{2}|\Tc_n(\hat{P}_i)|$ is an upper bound on the number of iterations in expurgation method for the type-by-type coding.

Now substituting \eqref{eqn:convenient_upperbound_type} in \eqref{eqn:bound_actual_seq_Type} from Lemma \ref {Lem:expurgation_Type} we find an upper bound on the error probability of every sequence $\xb\in\Tc_n(\hat{P}_i)$ in the codebook $\Cc_{{\rm ex},i}$ as

\begin{align}
	p_e(\xb,\Cc_{{\rm ex},i})\leq\left(\frac{2k_i|\Pc_n(\Xc)|}{M}\sum\limits_{\barxb\in\Tc_n(\hat{P}_i)}\left(\sum\limits_{\yb\in\Ycn}P_{\Yb|\Xb}(\yb|\xb)\left(\frac{q(\barxb,\yb)}{q(\xb,\yb)}\right)^s\right)^{\frac{1}{\rho}}\right)^{\rho}\label{eqn:bound_sequence_expurgated}
\end{align}

Notice that using \eqref{eqn:bound_sequence_expurgated} and summing over all source sequences $\xb\in\Tc_n(\hat{P}_i)$ we can find an upper bound that depends on the type $\hat{P}_i$. Therefore, resulting upper bound on the error probability of the codebook $\Cc$ requires maximization over the type. In order to find a simpler bound we can weaken the bound in \eqref{eqn:bound_sequence_expurgated} by including all $\barxb$ in the sum, however to keep the bound tight we introduce ratio of an arbitrary cost function as $\frac{e^{a(\barxb)}}{e^{a(\xb)}}$ where the cost function $a(\xb)$ depends on the sequence $\xb$ only through its type. Observe that the ratio is equal to $1$ when both $\xb$ and $\barxb$ have the same type, also this cost function can be optimized for to obtain the tightest bound. By upper-bounding $k_i$ by $k=n\log_{2}|\Xc|$ we obtain
\begin{align}
	p_e(\Ccb_{\rm ex})&=\sum\limits_{i=1}^{|\Pc_n(\Xc)|}\sum\limits_{\xb\in\Tc_n(\hat{P}_i)}P_{\Xb}(\xb)p_e(\xb,\Cc_{{\rm ex},i})\\&\leq\left(\frac{2k|\Pc_n(\Xc)|}{M}\right)^{\rho}\sum\limits_{\xb\in\Xcn}P_{\Xb}(\xb)\left(\sum\limits_{\barxb\in\Xcn}\left(\sum\limits_{\yb\in\Ycn}P_{\Yb|\Xb}(\yb|\xb)\frac{e^{a(\barxb)}}{e^{a(\xb)}}\left(\frac{q(\barxb,\yb)}{q(\xb,\yb)}\right)^s\right)^{\frac{1}{\rho}}\right)^{\rho}.\label{eqn:bound_expurgated_Type}
\end{align}

Equation \eqref{eqn:bound_expurgated_Type} is valid for any discrete source and any decoding metric. We now specialize this to the case of memoryless sources, metrics, and cost functions using $q(\xb,\yb)=\prod_{i=1}^{n}q(x_i,y_i)$ and $a(\xb)=\sum\limits_{i=1}^{n}a(x_i)$. Therefore we obtain
\begin{align}
	p_e(\Ccb_{\rm ex})&\leq\left(\frac{2k|\Pc_n(\Xc)|}{M}\right)^{\rho}\left(\sum\limits_{x\in\Xc}P_{X}(x)\left(\sum\limits_{\bar{x}\in\Xc}\left(\sum\limits_{y\in\Yc}P_{Y|X}(y|x)\frac{e^{a(\bar{x})}}{e^{a(x)}}\left(\frac{q(\bar{x},y)}{q(x,y)}\right)^s\right)^{\frac{1}{\rho}}\right)^{\rho}\right)^{n}\\
	&=e^{-n(\rho (R-\delta_n)-E^{\rm tt}_{\rm{x}}(\rho,s,a))}
\end{align}
where $\delta_n=\frac{\log2k|\Pc_n(\Xc)|}{n}\to 0$ as $n\to\infty$ and 
\begin{equation}
	E^{\rm tt}_{\rm{x}}(\rho,s,a(\cdot))=\log\sum\limits_{x\in\Xc}P_{X}(x)\left(\sum\limits_{\bar{x}\in\Xc}\left(\sum\limits_{y\in\Yc}P_{Y|X}(y|x)\frac{e^{a(\bar{x})}}{e^{a(x)}}\left(\frac{q(\bar{x},y)}{q(x,y)}\right)^s\right)^{\frac{1}{\rho}}\right)^{\rho}.\label{eqn:Ex_Type}
\end{equation}
Hence, we show the achievability of the exponent $E_{q,\mathrm{ex}}(R)$ by optimizing over $\rho,s,a(\cdot)$ as
\begin{equation}
	E^{\rm tt}_{q,\mathrm{ex}}(R)=\max_{\rho\geq 1,s\geq 0,a(\cdot)}\rho R-E_{\rm{x}}(\rho,s,a(\cdot)).\label{eqn:expurgated_exponent_Type}
\end{equation}

To show that the same code of Lemma \ref{Thm:Exponent_Type} achieves $E^{\rm tt}_{q,\mathrm{r}}(R)$ we use \eqref{eqn:bound_actual_seq_Type} of Lemma \ref{Lem:expurgation_Type} with $\rho=1$, and obtain that for every source sequence $\xb\in\Tc_n(\hat{P}_i)$
\begin{align}
	p_e(\xb,\Ccb_{\rm ex})\leq2\mathbb{E}\left[P_e(\xb,\Csf_i)\right].\label{eqn:special_case}
\end{align}

Averaging over all source sequences similar to the proof of Theorem \ref{th:standard} we obtain
\begin{align}
	p_e(\Ccb_{\rm ex})\leq2\mathbb{E}\left[p_e(\Csfb)\right],\label{eqn:ensemble_average_tbt}
\end{align}
which shows that the error probability of $\Ccb_{\rm ex}$ is upper bounded by twice the average error probability over type-by-type ensemble.

Now we derive an upper bound on the type-by-type ensemble average error probability following similar steps as Section \ref{sec:proof_th1}.

We start by bounding $p_e(\yb,\Ccb)$ for a given side information sequence $\yb$ and given type-by-type code $\Ccb$ as follows

\begin{align}
	p_e(\yb,\Ccb)&=\sum\limits_{i=1}^{|\Pc_n(\Xc)|}\sum\limits_{\xb\in\Tc_n(\hat{P}_i)}P_{\Xb|\Yb}(\xb|\yb)\mathds{1}\left[\bigcup_{\substack{\barxb\in\Tc_n(\hat{P}_i) \\ \barxb\neq\xb}}\left\{q(\barxb,\yb)\geq q(\xb,\yb) , \phi_i(\barxb)=\phi_i(\xb)\right\}\right]\\
	&\leq\sum\limits_{i=1}^{|\Pc_n(\Xc)|}\sum\limits_{\xb\in\Tc_n(\hat{P}_i)}P_{\Xb|\Yb}(\xb|\yb)\left(\sum_{\substack{\barxb\in\Tc_n(\hat{P}_i) \\ \barxb\neq\xb}}\mathds{1}\left[q(\barxb,\yb)\geq q(\xb,\yb) , \phi_i(\barxb)=\phi_i(\xb)\right]\right)^\rho\label{eqn:use_inequality_tbt}\\
	&\leq\sum\limits_{i=1}^{|\Pc_n(\Xc)|}\sum\limits_{\xb\in\Tc_n(\hat{P}_i)}P_{\Xb|\Yb}(\xb|\yb)\left(\sum_{\substack{\barxb\in\Tc_n(\hat{P}_i) \\ \barxb\neq\xb}}\left(\frac{q(\barxb,\yb)}{q(\xb,\yb)}\right)^{s}\mathds{1}\left[ \phi_i(\barxb)=\phi_i(\xb)\right]\right)^\rho,\label{eqn:bound_with_s_tbt}
\end{align}
where \eqref{eqn:use_inequality_tbt} follows from using the inequality $\mathds{1}\left[\bigcup\limits_{i}A_i\right]\leq\left(\sum\limits_{i}\mathds{1}\left[A_{i}\right]\right)^{\rho}$ for any set of events $\{A_i\}$ and $\rho\in[0,1]$ and \eqref{eqn:bound_with_s_tbt} holds for any $s\geq0$.

Now considering the ensemble of random type-by-type block source codes we upper bound the $\mathbb{E}\left[p_e(\yb,\Csfb)\right]$ using \eqref{eqn:bound_with_s_tbt} as follows
\begin{align}
	\mathbb{E}\left[p_e(\yb,\Csfb)\right]&\leq\sum\limits_{i=1}^{|\Pc_n(\Xc)|}\sum\limits_{\xb\in\Tc_n(\hat{P}_i)}P_{\Xb|\Yb}(\xb|\yb)\mathbb{E}\left[\left(\sum_{\substack{\barxb\in\Tc_n(\hat{P}_i) \\ \barxb\neq\xb}}\left(\frac{q(\barxb,\yb)}{q(\xb,\yb)}\right)^{s}\mathds{1}\left[ \Phi_i(\barxb)=\Phi_i(\xb)\right]\right)^\rho\right]\\
	&\leq\sum\limits_{i=1}^{|\Pc_n(\Xc)|}\sum\limits_{\xb\in\Tc_n(\hat{P}_i)}P_{\Xb|\Yb}(\xb|\yb)\left(\sum_{\substack{\barxb\in\Tc_n(\hat{P}_i) \\ \barxb\neq\xb}}\left(\frac{q(\barxb,\yb)}{q(\xb,\yb)}\right)^{s}\mathbb{E}\left[\mathds{1}\left[ \Phi_i(\barxb)=\Phi_i(\xb)\right]\right] \right)^\rho\label{eqn:jensen_ineq_tbt}\\
	&\leq\left(\frac{k_i|\Pc_n(\Xc)|}{M}\right)^{\rho}\sum\limits_{i=1}^{|\Pc_n(\Xc)|}\sum\limits_{\xb\in\Tc_n(\hat{P}_i)}P_{\Xb|\Yb}(\xb|\yb)\left(\sum_{\barxb\in\Tc_n(\hat{P}_i)}\left(\frac{q(\barxb,\yb)}{q(\xb,\yb)}\right)^{s}\right)^\rho\label{eqn:bound_fixed_side_information_tbt},\\
	&\leq\left(\frac{k|\Pc_n(\Xc)|}{M}\right)^{\rho}\sum\limits_{\xb\in\Xcn}P_{\Xb|\Yb}(\xb|\yb)\left(\sum_{\barxb\in\Xcn}\frac{e^{a(\barxb)}}{e^{a(\xb)}}\left(\frac{q(\barxb,\yb)}{q(\xb,\yb)}\right)^{s}\right)^\rho\label{eqn:bound_itroduce_cost_tbt},
\end{align}
where \eqref{eqn:jensen_ineq_tbt} follows from Jensen's inequality and the concavity of $x^{\rho}$ for $\rho\in\left[0,1\right]$ and \eqref{eqn:bound_fixed_side_information_tbt} follows from $\mathbb{E}\left[\mathds{1}\left[ \Phi(\barxb)=\Phi(\xb)\right]\right]\leq\frac{k}{M}$ and \eqref{eqn:bound_itroduce_cost_tbt} follows from weakening the bound by upper bounding $k_i$ with $k=n\log_{2}|\Xc|$ and including all $\barxb$ in the sum, however to keep the bound tight we introduce ratio of an arbitrary cost function as $\frac{e^{a(\barxb)}}{e^{a(\xb)}}$ where the cost function depends on the sequence $\xb$ only through its type..

Averaging over all side information sequences we obtain an upper bound on the type-by-type ensemble average error probability as
\begin{align}
	\mathbb{E}\left[p_e(\Csfb)\right]&=\sum\limits_{\yb\in\Ycn}P_{\Yb}(\yb)\mathbb{E}\left[p_e(\yb,\Csfb)\right]\\
	&\leq\left(\frac{k|\Pc_n(\Xc)|}{M}\right)^{\rho}\sum\limits_{\xb\in\Xcn,\yb\in\Ycn}P_{\Xb\Yb}(\xb,\yb)\left(\sum_{\barxb\in\Xcn}\frac{e^{a(\barxb)}}{e^{a(\xb)}}\left(\frac{q(\barxb,\yb)}{q(\xb,\yb)}\right)^{s}\right)^\rho\label{eqn:random_coding_bound_Type}
\end{align}
where \eqref{eqn:random_coding_bound_Type} holds for any $\rho\in\left[0,1\right]$ and $s\geq0$. Introducing \eqref{eqn:random_coding_bound_Type} in \eqref{eqn:ensemble_average_tbt} and particularizing it to the case of memoryless sources, metrics and cost functions we obtain
\begin{align}
	p_e(\Ccb_{\rm ex})&\leq2\left(\frac{k|\Pc_n(\Xc)|}{M}\right)^{\rho}\left(\sum\limits_{x\in\Xc,y\in\Yc}P_{XY}(x,y)\left(\sum\limits_{\bar{x}\in\Xc}\frac{e^{a(\bar{x})}}{e^{a(x)}}\left(\frac{q(\bar{x},y)}{q(x,y)}\right)^s\right)^{\rho}\right)^{n}\\
	&=e^{-n(\rho (R-\delta_n)-E^{\rm tt}_{\rm{s}}(\rho,s,a)-\delta'_n)}
\end{align}
where $\delta_n=\frac{\log k|\Pc_n(\Xc)|}{n}\to 0$ and $\delta'_n=\frac{\log 2}{n}\to 0$ as $n\to\infty$ and 
\begin{equation}
	E^{\rm tt}_{\rm{s}}(\rho,s,a(\cdot))=\log\sum\limits_{x\in\Xc,y\in\Yc}P_{XY}(x,y)\left(\sum\limits_{\bar{x}\in\Xc}\frac{e^{a(\bar{x})}}{e^{a(x)}}\left(\frac{q(\bar{x},y)}{q(x,y)}\right)^s\right)^{\rho}
\end{equation}
Hence, we show the achievability of the exponent $E^{\rm tt}_{q,\mathrm{r}}(R)$ by optimizing over $\rho,s,a(\cdot)$ as in \eqref{eqn:exp_Type_RC}.

\section{Proof of Proposition \ref{pr:equaiv}: Primal-dual Equivalence}
\label{sec:proof_primal_dual}

We start by showing the equivalence of \eqref{eq:RC_CK_SW_exponent} to \eqref{eqn:exp_Type_RC}. We rewrite \eqref{eq:RC_CK_SW_exponent} as
\begin{equation}
	E^{\mathrm{ck}}_{q,\mathrm{r}}(R)= \min_{Q}\left[D(Q\|P_{X})+ E_{q,\mathrm{r}}^{\mathrm{cc}}\left(H(Q)-R,Q,P_{Y|X}\right)\right],\label{eqn:primal_form_rc_exponent}
\end{equation}
\begin{equation}
	E_{q,\mathrm{r}}^{\mathrm{cc}}(H(Q)-R,Q,P_{Y|X})=\min_{P_{\hat{X}\tilde{X}\tilde{Y}}\in\mathcal{T}(Q)}\left[D(P_{\tilde{Y}|\tilde{X}}\|P_{Y|X}|P_{\tilde{X}})+|I(\tilde{Y};\hat{X})+R-H(Q)|^{+}\right],\label{eqn:Channel_CC_rc_exponent}
\end{equation}
where we use $P_{\tilde{X}}=Q$ and the identity
\begin{equation}
	D(P_{\tilde{X}\tilde{Y}}\|P_{XY})=D(P_{\tilde{X}}\|P_{X})+D(P_{\tilde{Y}|\tilde{X}}\|P_{Y|X}|P_{\tilde{X}}).\label{eqn:identity}
\end{equation}
to split the minimization in \eqref{eq:RC_CK_SW_exponent}.

Defining 
\begin{equation}
	\Sc(Q)=\left\lbrace P_{\tilde{X}\tilde{Y}}\in\Pc(\Xc\times\Yc):P_{\tilde{X}}=Q\right\rbrace,
\end{equation}
\begin{equation}
	\Tc(P_{\tilde{X}\tilde{Y}},Q)=\left\lbrace \tilde{P}_{\hat{X}\tilde{Y}}\in\Pc(\Xc\times\Yc):\tilde{P}_{\hat{X}}=Q,\tilde{P}_{\tilde{Y}}=P_{\tilde{Y}},\mathbb{E}_{\tilde{P}}\left[\log q(\hat{X},\tilde{Y})\right]\geq\mathbb{E}_{P}\left[\log q(\tilde{X},\tilde{Y})\right]\right\rbrace,
\end{equation}
we can rewrite \eqref{eqn:Channel_CC_rc_exponent} as
\begin{equation}
	E_{q,\mathrm{r}}^{\mathrm{cc}}(H(Q)-R,Q,P_{Y|X})=\min_{P_{\tilde{X}\tilde{Y}}\in\mathcal{S}(Q)}\min_{\tilde{P}_{\hat{X}\tilde{Y}}\in\Tc(P_{\tilde{X}\tilde{Y}},Q)}\left[D(P_{\tilde{Y}|\tilde{X}}\|P_{Y|X}|P_{\tilde{X}})+|I(\tilde{Y};\hat{X})+R-H(Q)|^{+}\right].\label{eqn:Channel_CC_rc_exp_break}
\end{equation}

The dual form of the constant composition random coding error exponent in \eqref{eqn:Channel_CC_rc_exp_break} is derived in \cite{scarlett2014_1}, using which we rewrite it as
\begin{align}
	E_{q,\mathrm{r}}^{\mathrm{cc}}(H(Q)-R,Q,P_{Y|X})=\sup_{\rho\in\left[0,1\right]}E_{0}^{\mathrm{cc}}(Q,\rho)-\rho(H(Q)-R),
\end{align}
where
\begin{align}
	E_{0}^{\mathrm{cc}}(Q,\rho)=\sup_{s\geq 0,b(\cdot)}-\sum\limits_{x\in\Xc}Q(x)\log\sum\limits_{y\in\Yc}P_{Y|X}(y|x)\left(\sum\limits_{\bar{x}\in\Xc}Q(\bar{x})\frac{e^{b(\bar{x})}}{e^{b(x)}}\left(\frac{q(\bar{x},y)}{q(x,y)}\right)^s\right)^{\rho}.\label{eqn:RC_cc_channel}
\end{align}
Defining $e^{a(x)}=Q(x)e^{b(x)}$, we can rewrite \eqref{eqn:Ex_cc_channel} as
\begin{align}
	E_{0}^{\mathrm{cc}}(Q,\rho)&=\sup_{s\geq 0,a(\cdot)}-\sum\limits_{x\in\Xc}Q(x)\log\sum\limits_{y\in\Yc}P_{Y|X}(y|x)\left(\sum\limits_{\bar{x}\in\Xc}Q(x)\frac{e^{a(\bar{x})}}{e^{a(x)}}\left(\frac{q(\bar{x},y)}{q(x,y)}\right)^s\right)^{\rho}\\
	&=\rho H(Q)+\sup_{s\geq 0,a(\cdot)}-\sum\limits_{x\in\Xc}Q(x)\log\sum\limits_{y\in\Yc}P_{Y|X}(y|x)\left(\sum\limits_{\bar{x}\in\Xc}\frac{e^{a(\bar{x})}}{e^{a(x)}}\left(\frac{q(\bar{x},y)}{q(x,y)}\right)^s\right)^{\rho}\label{eqn:RC_cc_rewrite}
\end{align}

Introducing \eqref{eqn:RC_cc_rewrite} in \eqref{eqn:RC_cc_channel}, noting that the objective in \eqref{eqn:RC_cc_channel} is concave in $Q$ and using Fan's minimax theorem interchanging the minimum and supremum we obtain
\begin{align}
	E_{\rm{ex}}(R)=\sup_{\rho\geq 1}\left[\rho R+\sup_{s\geq 0,a(\cdot)}\min_{Q}L\right],\label{eqn:RC_exp_primal}
\end{align}
where 
\begin{equation}
	L=\sum\limits_{x\in\Xc}Q(x)\log\frac{Q(x)}{P_{X}(x)}-\sum\limits_{x\in\Xc}Q(x)\log\sum\limits_{y\in\Yc}P_{Y|X}(y|x)\left(\sum\limits_{\bar{x}\in\Xc}\frac{e^{a(\bar{x})}}{e^{a(x)}}\left(\frac{q(\bar{x},y)}{q(x,y)}\right)^s\right)^{\rho}.\label{eqn:minimization_objective_rc}
\end{equation}

Setting $\frac{\partial L}{\partial Q}=0$ we find the minimizing distribution as the following mismatched tilted distribution
\begin{equation}
	Q(x)=CP_{X}(x)\sum\limits_{y\in\Yc}P_{Y|X}(y|x)\left(\sum\limits_{\bar{x}\in\Xc}\frac{e^{a(\bar{x})}}{e^{a(x)}}\left(\frac{q(\bar{x},y)}{q(x,y)}\right)^s\right)^{\rho}\label{eqn:minimizing_distribution_rc}
\end{equation}
where $C$ is the normalization constant given by
\begin{equation}
	C^{-1}=\sum\limits_{x\in\Xc}P_{X}(x)\sum\limits_{y\in\Yc}P_{Y|X}(y|x)\left(\sum\limits_{\bar{x}\in\Xc}\frac{e^{a(\bar{x})}}{e^{a(x)}}\left(\frac{q(\bar{x},y)}{q(x,y)}\right)^s\right)^{\rho}.
\end{equation}

From \eqref{eqn:minimizing_distribution_rc} we have $\frac{Q(x)}{CP_{X}(x)}=\sum\limits_{y\in\Yc}P_{Y|X}(y|x)\left(\sum\limits_{\bar{x}\in\Xc}\frac{e^{a(\bar{x})}}{e^{a(x)}}\left(\frac{q(\bar{x},y)}{q(x,y)}\right)^s\right)^{\rho}$, replacing it in \eqref{eqn:minimization_objective_rc} we obtain
\begin{align}
	\min_{Q}L=-\log \sum\limits_{x\in\Xc}P_{X}(x)\sum\limits_{y\in\Yc}P_{Y|X}(y|x)\left(\sum\limits_{\bar{x}\in\Xc}\frac{e^{a(\bar{x})}}{e^{a(x)}}\left(\frac{q(\bar{x},y)}{q(x,y)}\right)^s\right)^{\rho}\label{eqn:min_L_rc}
\end{align}
Plugging \eqref{eqn:min_L_rc} into \eqref{eqn:RC_exp_primal} we obtain the dual form of the type-by-type random coding exponent in \eqref{eqn:exp_Type_RC}.

To show the equivalence of \eqref{eq:EX_CK_SW_exponent} to \eqref{eqn:exp_Type_EX}. We rewrite \eqref{eq:EX_CK_SW_exponent} as

\begin{equation}
	E^{\mathrm{ck}}_{q,\rm ex}(R) = \min_{Q}\left[D(Q\|P_{X})+ E_{q,\mathrm{ex}}^{\mathrm{cc}}\left(H(Q)-R,Q,P_{Y|X}\right)\right],\label{eqn:primal_form_expurgated_exponent}
\end{equation}
where
\begin{equation}
	E_{q,\mathrm{ex}}^{\mathrm{cc}}(H(Q)-R,Q,P_{Y|X})=\min_{\substack{P_{\hat{X}\tilde{X}\tilde{Y}}\in\mathcal{T}(Q)\\H(\hat{X}|\tilde{X})\geq R}}\left[D(P_{\tilde{Y}|\tilde{X}}\|P_{Y|X}|P_{\tilde{X}})+I(\tilde{X}\tilde{Y};\hat{X})+R-H(Q)\right]\label{eqn:Channel_CC_expurgated_exponent}
\end{equation}
and
\begin{equation}
	\Tc(Q)=\left\lbrace P_{\hat{X}\tilde{X}\tilde{Y}}\in\Pc(\Xc\times\Xc\times\Yc):P_{\hat{X}}=P_{\tilde{X}}=Q,\mathbb{E}\left[\log q(\hat{X},\tilde{Y})\right]\geq\mathbb{E}\left[\log q(\tilde{X},\tilde{Y})\right]\right\rbrace,
\end{equation}
where we use $P_{\tilde{X}}=Q$ and the identity \eqref{eqn:identity} to split the minimization in \eqref{eq:EX_CK_SW_exponent}.

Defining 
\begin{equation}
	\Sc(Q)=\left\lbrace \tilde{P}_{\hat{X}\tilde{X}}\in\Pc(\Xc\times\Xc):\tilde{P}_{\hat{X}}=\tilde{P}_{\tilde{X}}=Q\right\rbrace,
\end{equation}
\begin{equation}
	\Tc(\tilde{P}_{\hat{X}\tilde{X}})=\left\lbrace P_{\hat{X}\tilde{X}\tilde{Y}}\in\Pc(\Xc\times\Xc\times\Yc):P_{\hat{X}\tilde{X}}=\tilde{P}_{\hat{X}\tilde{X}},\mathbb{E}\left[\log q(\hat{X},\tilde{Y})\right]\geq\mathbb{E}\left[\log q(\tilde{X},\tilde{Y})\right]\right\rbrace,
\end{equation}
we can rewrite \eqref{eqn:Channel_CC_expurgated_exponent} as
\begin{equation}
	E_{q,\mathrm{ex}}^{\mathrm{cc}}(H(Q)-R,Q,P_{Y|X})=\min_{\substack{\tilde{P}_{\hat{X}\tilde{X}}\in\mathcal{S}(Q)\\H(\hat{X}|\tilde{X})\geq R}}\min_{P_{\hat{X}\tilde{X}\tilde{Y}}\in\Tc(\tilde{P}_{\hat{X}\tilde{X}})}\left[D(P_{\hat{X}\tilde{X}\tilde{Y}}\|\tilde{P}_{\hat{X}\tilde{X}}\times P_{Y|X})+R-H(\hat{X}|\tilde{X})\right]\label{eqn:Channel_CC_expurgated_exp_broken}
\end{equation}
where we use \cite[Eq. (32)]{Csi81} as
\begin{equation}
	D(P_{\tilde{Y}|\tilde{X}}\|P_{Y|X}|P_{\tilde{X}})+I(\tilde{X}\tilde{Y};\hat{X})=D(P_{\hat{X}\tilde{X}\tilde{Y}}\|\tilde{P}_{\hat{X}\tilde{X}}\times P_{Y|X})+I(\tilde{X};\hat{X}).
\end{equation}

For a given $\tilde{P}_{\hat{X}\tilde{X}}\in\Sc(Q)$, $R-H(\hat{X}|\tilde{X})$ is constant, therefore we consider the optimization problem
\begin{equation}
	\min_{P_{\hat{X}\tilde{X}\tilde{Y}}\in\Tc(\tilde{P}_{\hat{X}\tilde{X}})}D(P_{\hat{X}\tilde{X}\tilde{Y}}\|\tilde{P}_{\hat{X}\tilde{X}}\times P_{Y|X}).
\end{equation}

The dual of this optimization problem has been found in \cite[Appendix B]{scarlett2014}. Using very similar arguments we find the following equivalent forms of \eqref{eqn:Channel_CC_expurgated_exp_broken} as
\begin{align}
	E_{q,\mathrm{ex}}^{\mathrm{cc}}(H(Q)-R,Q,P_{Y|X})&=\sup_{s\geq0}\min_{\substack{P_{X\bar{X}}:P_{X}=Q,P_{\bar{X}}=Q\\H(\bar{X}|X)\geq R}}\mathbb{E}_{P}\left[d_{s}(X,\bar{X})\right]+R-H(\bar{X}|X),\\
	&=\sup_{\rho\geq1}E_{x}^{\mathrm{cc}}(Q,\rho)-\rho(H(Q)-R)\label{eqn:lagrange_dual},
\end{align}
where
\begin{equation}
	d_s(x,\bar{x})=-\log\sum\limits_{y}P_{Y|X}(y|x)\left(\frac{q(\bar{x},y)}{q(x,y)}\right)^{s},
\end{equation}
and
\begin{align}
	E_{\rm{x}}^{\mathrm{cc}}(Q,\rho)=\sup_{s\geq 0,b(\cdot)}-\rho\sum\limits_{x\in\Xc}Q(x)\log\sum\limits_{\bar{x}\in\Xc}Q(\bar{x})\left(\sum\limits_{y\in\Yc}P_{Y|X}(y|x)\frac{e^{b(\bar{x})}}{e^{b(x)}}\left(\frac{q(\bar{x},y)}{q(x,y)}\right)^s\right)^{\frac{1}{\rho}}.\label{eqn:Ex_cc_channel}
\end{align}

Using \eqref{eqn:lagrange_dual} in \eqref{eqn:primal_form_expurgated_exponent} we obtain
\begin{align}
	E^{\mathrm{ck}}_{\rm{ex}}(R)=\min_{Q}\left[ D(Q\|P_{X})+\sup_{\rho\geq 1}\left[E_{x}^{\mathrm{cc}}(Q,\rho)-\rho(H(Q)-R)\right] \right],\label{eqn:Expurgated_exp_primal}
\end{align}
Defining $e^{a(x)}=Q(x)^{\rho}e^{b(x)}$, we can rewrite \eqref{eqn:Ex_cc_channel} as
\begin{align}
	E_{\rm{x}}^{\mathrm{cc}}(Q,\rho)&=\sup_{s\geq 0,a(\cdot)}-\rho\sum\limits_{x\in\Xc}Q(x)\log\sum\limits_{\bar{x}\in\Xc}Q(x)\left(\sum\limits_{y\in\Yc}P_{Y|X}(y|x)\frac{e^{a(\bar{x})}}{e^{a(x)}}\left(\frac{q(\bar{x},y)}{q(x,y)}\right)^s\right)^{\frac{1}{\rho}}\\
	&=\rho H(Q)+\sup_{s\geq 0,a(\cdot)}-\rho\sum\limits_{x\in\Xc}Q(x)\log\sum\limits_{\bar{x}\in\Xc}\left(\sum\limits_{y\in\Yc}P_{Y|X}(y|x)\frac{e^{a(\bar{x})}}{e^{a(x)}}\left(\frac{q(\bar{x},y)}{q(x,y)}\right)^s\right)^{\frac{1}{\rho}}\label{eqn:Ex_cc_rewrite}
\end{align}

Introducing \eqref{eqn:Ex_cc_rewrite} in \eqref{eqn:Expurgated_exp_primal}, noting that the objective in \eqref{eqn:Expurgated_exp_primal} is concave in $Q$ and using Fan's minimax theorem interchanging the minimum and supremum we obtain
\begin{align}
	E^{\mathrm{ck}}_{\rm{ex}}(R)=\sup_{\rho\geq 1}\left[\rho R+\sup_{s\geq 0,a(\cdot)}\min_{Q}L\right],\label{eqn:exp_dual}
\end{align}
where 
\begin{equation}
	L=\sum\limits_{x\in\Xc}Q(x)\log\frac{Q(x)}{P_{X}(x)}-\rho\sum\limits_{x\in\Xc}Q(x)\log\sum\limits_{\bar{x}\in\Xc}\left(\sum\limits_{y\in\Yc}P_{Y|X}(y|x)\frac{e^{a(\bar{x})}}{e^{a(x)}}\left(\frac{q(\bar{x},y)}{q(x,y)}\right)^s\right)^{\frac{1}{\rho}}\label{eqn:minimization_objective}
\end{equation}
Setting $\frac{\partial L}{\partial Q}=0$ we find the minimizing distribution as the following mismatched tilted distribution
\begin{equation}
	Q(x)=CP_{X}(x)\left(\sum\limits_{\bar{x}\in\Xc}\left(\sum\limits_{y\in\Yc}P_{Y|X}(y|x)\frac{e^{a(\bar{x})}}{e^{a(x)}}\left(\frac{q(\bar{x},y)}{q(x,y)}\right)^s\right)^{\frac{1}{\rho}}\right)^\rho\label{eqn:minimizing_distribution}
\end{equation}
where $C$ is the normalization constant given by
\begin{equation}
	C^{-1}=\sum\limits_{x\in\Xc}P_{X}(x)\left(\sum\limits_{\bar{x}\in\Xc}\left(\sum\limits_{y\in\Yc}P_{Y|X}(y|x)\frac{e^{a(\bar{x})}}{e^{a(x)}}\left(\frac{q(\bar{x},y)}{q(x,y)}\right)^s\right)^{\frac{1}{\rho}}\right)^\rho.
\end{equation}

From \eqref{eqn:minimizing_distribution} we have $\frac{Q(x)}{CP_{X}(x)}=\left(\sum\limits_{\bar{x}\in\Xc}\left(\sum\limits_{y\in\Yc}P_{Y|X}(y|x)\frac{e^{a(\bar{x})}}{e^{a(x)}}\left(\frac{q(\bar{x},y)}{q(x,y)}\right)^s\right)^{\frac{1}{\rho}}\right)^\rho$, replacing it in \eqref{eqn:minimization_objective} we obtain
\begin{align}
	\min_{Q}L=-\log \sum\limits_{x\in\Xc}P_{X}(x)\left(\sum\limits_{\bar{x}\in\Xc}\left(\sum\limits_{y\in\Yc}P_{Y|X}(y|x)\frac{e^{a(\bar{x})}}{e^{a(x)}}\left(\frac{q(\bar{x},y)}{q(x,y)}\right)^s\right)^{\frac{1}{\rho}}\right)^\rho\label{eqn:min_L}
\end{align}
Plugging \eqref{eqn:min_L} into \eqref{eqn:exp_dual} we obtain the dual form of the type-by-type expurgated exponent in \eqref{eqn:exp_Type_EX}.

\bibliographystyle{IEEEtran}
\bibliography{itc}

\end{document}